\documentclass[a4paper,accepted=2025-03-24,onecolumn]{quantumarticle}
\pdfoutput=1
\usepackage[utf8]{inputenc}
\usepackage[T1]{fontenc}
\usepackage{amsmath, amssymb, amsthm}
\usepackage{dsfont}
\usepackage[table,xcdraw]{xcolor}
\usepackage{graphicx}
\usepackage{multirow}
\usepackage{float}
\usepackage{listings}
\usepackage{wrapfig}
\usepackage{enumitem}
\usepackage{authblk}
\usepackage{slashed}
\usepackage{mleftright}
\mleftright
\usepackage{subcaption}

\usepackage{tikz}
\usetikzlibrary{matrix, shapes, arrows, chains, decorations.pathmorphing, decorations.pathreplacing, calligraphy}
\tikzset{quantum/.style={decorate, decoration=snake}}

\newcommand{\ket}[1]{|#1\rangle}							
\newcommand{\bra}[1]{\langle#1|}

\newcommand{\ketbra}[2]{|#1\rangle\langle#2|}
\newcommand{\braket}[2]{\langle #1\lvert#2\rangle}
\newcommand{\abs}[1]{\lvert #1\rvert}

\newcommand{\norm}[1]{\| #1\|}

\newcommand{\QPVBB}{$\mathrm{QPV}_{\mathrm{BB84}}$}
\newcommand{\QPVBBeta}{$\mathrm{QPV}_{\mathrm{BB84}}^{\eta}$}

\renewcommand{\Im}{\operatorname{Im}}

\newcommand{\tr}[1]{\mathrm{Tr}\left[#1\right]}
\newcommand{\ptr}[2]{\mathrm{Tr}_{#1}\left[#2\right]}

\theoremstyle{plain}
\newtheorem{theorem}{Theorem}[section]
\newtheorem{observation}[theorem]{Observation}
\newtheorem{remark}[theorem]{Remark}
\newtheorem{prop}[theorem]{Proposition}
\newtheorem{definition}[theorem]{Definition}
\newtheorem{ex}[theorem]{Example}
\newtheorem{lemma}[theorem]{Lemma}

\newtheorem{result}[theorem]{Result}

\usepackage[pdfencoding=auto, psdextra]{hyperref}
\hypersetup{
    colorlinks=true,
    linktoc=all,
    linkcolor=blue,
    citecolor=[rgb]{0.2, 0.4, 0.7},
}

\usepackage[title]{appendix}

\begin{document}

\title{Lossy-and-Constrained Extended Non-Local Games with Applications to Quantum Cryptography
}

\author[1,2]{Lloren\c{c} Escol\`a-Farr\`as}
\author[1,2]{Florian Speelman}
\newcommand{\lle}[1]{{\color{blue}#1}}
\newcommand{\fs}[1]{{\textcolor{red}{[Florian: #1]}}}

\affil[1]{QuSoft, CWI Amsterdam, Science Park 123, 1098 XG Amsterdam, The Netherlands}
\affil[2]{Multiscale Networked Systems (MNS),  Informatics Institute, University of Amsterdam, Science Park 904,  1098 XH Amsterdam, The Netherlands }

\maketitle
\vspace{-1em}
\begin{abstract}
Extended non-local games are a generalization of monogamy-of-entanglement games, played by two quantum parties and a quantum referee that performs a measurement on their local quantum system. Along the lines of the NPA hierarchy, the optimal winning probability of those games can be upper bounded by a hierarchy of semidefinite programs (SDPs) converging to the optimal value. Here, we show that if one extends such games by considering \emph{constraints} and \emph{loss}, motivated by experimental errors and loss through quantum communication, the convergence of the SDPs to the optimal value still holds. We give applications of this result, and we compute SDPs that show tighter security of protocols for relativistic bit commitment, quantum key distribution, and quantum position verification. 
\end{abstract}
\tableofcontents
\section{Introduction}

The study of quantum correlations attainable by distant parties, known as non-local correlations, has been a broad and well-studied topic in quantum information theory. They are not only interesting from a fundamental point of view, given that it is well-known that quantum distant parties can attain correlations that classical physics is not able to describe \cite{Bell_1964}, but also lead to many applications such as secure key distribution \cite{acin_device-independent_2007}, 
certified randomness \cite{pironio_random_2010}, 
reduced communication complexity \cite{buhrman_nonlocality_2010},
self-testing \cite{mayers_self_2004,supic_self-testing_2019},
and computation \cite{anders_computational_2009}. 

Non-local correlations are often studied in the literature as non-local games, which provide an operational framework for understanding them. Usually, a classical referee sends classical questions to the two distant parties, namely Alice and Bob, who have to respond classical answers. The game is defined according to a predicate that accepts the answers as correct if they fulfill a previously agreed relation with the questions. One of the most well-known examples in the literature is given by the CHSH game \cite{CHSHArticle}. 

The usage of semidefinite programming (SDP) techniques to bound correlations of non-local games was initiated by Cleve, H{\o}yer, Toner and  Watrous \cite{https://doi.org/10.48550/arxiv.quant-ph/0404076}, Wehner \cite{Wehner_2006}, and Liang and Doherty~\cite{PhysRevA.75.042103}. Later on, Navascués, Pironio and Acín introduced an infinite hierarchy of conditions, the so-called NPA hierarchy, necessarily satisfied by any set of quantum correlations, each of them testable with SDPs \cite{PhysRevLett.98.010401}, with the subsequent result showing that this hierarchy is complete \cite{NPA2008}. 

Non-local correlations have also been investigated beyond the classical questions and answers scenario, e.g.\ in a quantum XOR game~\cite{regev2012quantumxorgames}, a quantum referee who sends quantum questions to Alice and Bob, who have to respond with classical answers ---classical-quantum games---, or in a rank-one quantum game~\cite{cooney2013rankonequantumgames}, both the questions and the answers are quantum. 
In this work, we consider a slightly different scenario. In terms of games, the referee is quantum and performs a measurement on their local register of a tripartite system shared together with the two collaborative parties Alice and Bob. These type of games are known in the literature as extended non-local games, introduced by \cite{Extended_non-local_games_andMoE_games}. They are a generalization of monogamy-of-entanglement games \cite{TomamichelMonogamyGame2013},
where the referee, Alice and Bob pre-share a joint quantum state, the referee performs a local measurement on their local register, chosen from a given set of measurements, then, the referee announces the chosen measurement, and the task of Alice and Bob is to guess the outcome.
These types of games have applications to quantum cryptography tasks, such as quantum position verification and device-independent quantum key distribution \cite{TomamichelMonogamyGame2013}, and have been studied in the context of quantum steering~\cite{russo2017extendednonlocalgamesquantumclassical}. 
Using similar ideas as in \cite{NPA2008}, it was shown \cite{Extended_non-local_games_andMoE_games} that there exists a hierarchy of SDPs that bound the optimal winning probability of an extended non-local game, and it is such that converges to its optimal value. 

When analyzing the applications of such games, for example in quantum communication tasks, one has to take into account experimental parameters such as errors and loss of the quantum information, and treating these two separately can give significant improvements in the analysis, see e.g.~\cite{Escol_Farr_s_2023}, where a loss-tolerance of up to 50\% is shown for a given protocol with only around 15\% error-tolerance \cite{TomamichelMonogamyGame2013},  which is impossible if naively treating loss as error.
Moreover, besides errors and loss, when executing a protocol it can be the case that certain answers are expected to be observed with a given frequency, for example: a certain number of correct answers;  some wrong answers (caused e.g.~by errors); some `empty' answers, if the quantum information got lost when transmitted. Additionally, in some protocols there are answer combinations that are simply inconsistent with the setting, and are never expected to be observed: including such constraints in the analysis of extended non-local games could enable improved bounds in practical applications.\\

\textbf{Contributions.} We introduce a modification of extended non-local games, which we call \emph{lossy-and-constrained} extended non-local games, that takes into account errors, loss and the fact that certain answers are expected to be observed with a given frequency. 
These games are inspired by practical considerations, describing scenarios where honest parties receive quantum states over a lossy channel, and we aim to prevent security issues that can occur because of such transmission loss. 
We show that similar results as in \cite{Extended_non-local_games_andMoE_games} can be applied to the \emph{lossy-and-constrained} version, and that there exists a hierarchy of SDPs converging to the optimal value that can be attained by the quantum parties Alice and Bob. 

We consider different monogamy-of-entanglement and extended non-local games and we analyze them in their lossy-and-constrained version. By computing the corresponding semidefinite programs, we show a way to verify some previously known results by solely computing an SDP. Moreover, we find new tighter and new tight results for those games.  It is worth to highlight that for most of the games that we analyze with loss or constraints, we obtain tight values by using the first level of the hierarchy --- showing that resulting SDPs are numerically solvable in practice and going to higher levels is not required. We created Python codes to compute these values via SDP \cite{code}, using \emph{cvxpy} \cite{anders_computational_2009}. These programs can be modified and adapted to analyze other games.

Finally, we study applications to quantum cryptography, showing that (lossy-and-constrained) non-local games can be used to prove security for certain protocols in a simpler way, with a method that can executed using only the description of the game. Using our analysis, we provide tighter security bounds. 
First, we study a relativistic bit commitment protocol originally introduced by Kent \cite{Kent_2012_BC}, where $n$-BB84 states \cite{BB84} are used in the protocol. In \cite{Kaniewski_2013}, security was shown for large $n$, a result that was tightened in \cite{Pital_a_Garc_a_2018} that applied for every $n$. Here, we show that the security of this protocol can be reduced to an extended non-local game, and we provide tighter bounds by computing the corresponding SDPs for small values of~$n$.  Second, we use a constrained extended non-local game to show security of device-independent quantum key distribution \cite{BB84}, solving SDPs for small $n$, where $n$ is the number qubits used in the protocol. Third, we use lossy-and-constrained extended non-local games to study security for certain protocols in quantum position verification (QPV) in the presence of photon loss. For the  \QPVBB~protocol with loss,  we obtain, in an easier way, the results shown in \cite{Escol_Farr_s_2023}, where it was necessary to prove a relation between certain operator norms and error and then use it in the NPA hierarchy to obtain the values. The advantage is that by considering the corresponding game with loss and constraints, it is enough to have the description to compute the SDP providing the optimal values. We improve the bounds found in \cite{Escol_Farr_s_2023} for an extension of the above protocol, and we show that they are tight. We finally apply our results to show security of the 2-fold parallel repetition of  \QPVBB~with loss.

\section{Preliminaries}
For $n\in\mathbb N$, we will denote $[n]:=\{0,\ldots,n-1\}$. Let $\mathcal{H}$, $\mathcal{H'}$ be finite-dimensional Hilbert spaces, we denote by $\mathcal{B}(\mathcal{H},\mathcal{H'})$ the set of bounded operators from $\mathcal{H}$ to $\mathcal{H'}$ and $\mathcal{B}(\mathcal{H})=\mathcal{B}(\mathcal{H},\mathcal{H})$. For $A,B\in\mathcal{B}(\mathcal H)$, we denote $A\succeq B$ if $A-B$ is positive semidefinite. 
Denote by $\mathcal{S}(\mathcal{H})$ the set of quantum states on $\mathcal{H}$,~i.e.\ $\mathcal{S}(\mathcal{H})=\{\rho\in\mathcal{B}(\mathcal{H})\mid \rho\geq0, \tr{\rho}=1)\}$. A pure state will be denoted by a ket $\ket{\psi}\in\mathcal{H}$. A maximally entangled state or EPR pair is the state $\ket{\Phi^+}=\frac{1}{\sqrt{2}}(\ket{00}+\ket{11}).$ We will refer to basis 0 and 1 to denote the computational and Hadamard basis, respectively. The Hadamard transformation will be denoted by $H=\frac{1}{\sqrt{2}}\begin{pmatrix}
1 & 1 \\
1& -1 
\end{pmatrix}$.   For bit strings $x,a\in\{0,1\}^n$ we denote 
\begin{equation}
    \ketbra{a^x}{a^x}:=\ketbra{a_0^{x_0}}{a_0^{x_0}}\otimes\cdots\otimes\ketbra{a_{n-1}^{x_{n-1}}}{a_{n-1}^{x_{n-1}}},
\end{equation}
where $\ketbra{a_i^{x_i}}{a_i^{x_i}}=H^{x_i}\ketbra{a_i}{a_i}H^{x_i}$. 
The Hamming distance $d_{H}$ between two bit strings $x,y\in\{0,1\}^n$ is the number of positions at which they differ, i.e.
\begin{equation}
    d_{H}(x,y):=\{i\in[n]\mid x_i\neq y_i\}.
\end{equation}
For a set $\mathcal{X}$, we denote $\mathcal{X}^{\times n}:=\mathcal{X}\times\cdots\times\mathcal{X}$, where the Cartesian product is taken $n$ times.

\section{Lossy-and-constrained extended non-local games}
Following \cite{Extended_non-local_games_andMoE_games}, we describe extended non-local games. Let $\mathcal{Z}$, $\mathcal{X}$, $\mathcal{Y}$, $\mathcal{V}$, $\mathcal{A}$ and $\mathcal{B}$ be finite non-empty alphabets, let $q$ be a probability distribution over $\mathcal{Z}\times\mathcal{X}\times\mathcal{Y}$, let $V^{z}_{v}$ be a square Hermitian matrix of dimension $m\in\mathbb{N}$ for all ${(z,v)\in\mathcal{Z}\times\mathcal{V}}$ and let $V:=\{V^z_v\}_{z,v}$. Three parties will play a role in the extended non-local games: a referee, Alice and Bob, whose associated Hilbert spaces will be denoted by $\mathcal{ H}_R, \mathcal{ H}_A$ and $\mathcal{ H}_B$, respectively.

\begin{definition}
An \emph{extended non-local game} $\mathcal{G}$, played by a referee $R$ and two collaborative parties Alice ($A$) and Bob ($B$), denoted by the tuple
\begin{equation}
    \mathcal{G}=(q,V,W),
\end{equation}
where $W=\{(z,x,y,v,a,b)\mid pred(v,a,b|z,x,y)=1\}$ (winning set), for a certain predicate function $pred(v,a,b|z,x,y)\in\{0,1\}$, is described as follows: 
\begin{enumerate}
    \item Alice and Bob prepare a tripartite state $\rho_{RAB}$ where the dimension of the register $R$ is $m$, i.e.\ the reduced state $\rho_R\in\mathcal{S}(\mathcal{H}_R)$, for $\mathcal{H}_R\cong \mathbb C^m$.
    \item They send the register $R$ of the tripartite state to the referee. The two parties are no longer allowed to communicate. 
    \item The referee sends questions $x\in\mathcal{X}$ and $y\in\mathcal{Y}$ to Alice and Bob and picks $z\in\mathcal{Z}$ according to the probability distribution $q(z,x,y)$. Then, Alice and Bob have to answer $a\in\mathcal{A}$  and $b\in\mathcal{B}$, respectively. Denote by $\rho_{R_{a,b}}^{x,y}\in\mathcal{S}(\mathcal{H}_R)$ the resulting quantum state held by $R$ after Alice and Bob send $a$ and $b$. The average pay-off for the players is given by
    \begin{equation}
        \sum_{z,x,y,v,a,b\in W}q(z,x,y)\tr{V^{z}_{v}\rho_{R_{a,b}}^{x,y}}.
    \end{equation}
\end{enumerate}
\end{definition}
\noindent See Fig.~\ref{fig:game} for a schematic representation of an extended non-local game. 

\begin{figure}[h]
    \centering
    \begin{tikzpicture}[node distance=1cm, auto]
    \node (xV) {$z$};
    \node [right=3cm of xV] (emptyx) {};
    \node [right=0.4cm of emptyx] (x1) {$x$};
    \node [right=3cm of x1] (emptyx1) {};
    \node [right=0.4cm of emptyx1] (x2) {$y$};
    %\node [right=of x2] (emptyx2) {};
    %\node [right=of emptyx2] (x3) {$x$};
    
    \node [below=0.9 of xV] (belowxV){};
    \node [below=0.3cm of xV] (belowxVbefore){};
    \node [left=0.0 of belowxV] (leftbelowxV){};
    \node [above=0.1 of belowxV] (aboveleftbelowxV){$R$};
    \node [below=0.4 of aboveleftbelowxV] (aboveleftbelowxV_meas){};
    \node [below=0.3 of aboveleftbelowxV] (aboveleftbelowxV_V){$\{V_v^x\}_v$};
    \node [below=0.3cm of xV] (belowxVbefore){};
    \node [below=1.6cm of xV] (Vans1){};
    \node [below=2.2cm of xV] (Vans2){};
    \node [below=2.3cm of xV] (Vans3){$v$};
    \draw [->] (Vans1) -- (Vans2);

    \node [below=of emptyx] (belowemptyx){};
    \node [below=0.9cm of x1] (belowx1){};
    \node [left=0.0 of belowx1] (leftbelowx1){};
    \node [below=0.3 of x1] (belowemptyx1){$A$};
    \node [below=0.3 of belowemptyx1] (MeasAlice){$\{A_a^x\}_a$};
    \node [above=0.1 of leftbelowx1] (aboveleftbelowx1){};
    \node [below=0.3cm of x1] (belowx1before){};
    \node [below=1.6cm of x1] (Aans1){};
    \node [below=2.2cm of x1] (Aans2){};
    \node [below=2.3cm of x1] (Aans3){$a$};
    \draw [->] (Aans1) -- (Aans2);

    %\node [below=1.2cm of emptyx1] (belowemptyx1){$\ldots$};
    \node [below=0.9cm of  x2] (belowx2){};
    \node [left=0.5 of belowx2] (leftbelowx2){};
    \node [below=0.3 of x2] (belowemptyx2){$B$};
    \node [below=0.3 of belowemptyx2] (MeasBob){$\{B_b^x\}_b$};
    \node [above=0.1 of leftbelowx2] (aboveleftbelowx2){};
    \node [below=0.3cm of x2] (belowx2before){};
    \node [below=1.6cm of x2] (Bans1){};
    \node [below=2.2cm of x2] (Bans2){};
    \node [below=2.3cm of x2] (Bans3){$b$};
    \draw [->] (Bans1) -- (Bans2);

    %\node [below=1.2 cmof emptyx2] (belowemptyx2){$\ldots$};
    %\node [below=of x3] (belowx3){$A_{k-1}$};
    %\node [below=0.6cm of x3] (belowx3before){};

     \node [below=0.6cm of belowxV] (AlicebelowxV){};
     \node [below=0.6cm of belowx1] (Alicebelowx1){};
     %\node [below=2.2cm of emptyx] (AliceState){$\mathop{\otimes}\limits_{x_i:x_i=0}\ket{\Phi^+}_{R_{x_i}A_{x_i}}$};
     
     %\draw [-, transform canvas={xshift=0pt, yshift = 0 pt}, quantum] (AlicebelowxV) -- (Alicebelowx1) {};
     
     \node [below=1.7cm of belowxV] (2belowxV){};
     \node [below=1.7cm of belowx2] (2belowx2){};
     %\node [below=2cm of belowx1] (2belowx1){$\mathop{\otimes}\limits_{x_i:x_i=1}\ket{\Phi^+}_{R_{x_i}B_{x_i}}$};

     \node[mark size=2.5pt,color=black] at (belowxV) {\pgfuseplotmark{*}};
     \node[mark size=2.5pt,color=black] at (belowx1) {\pgfuseplotmark{*}};
     \node[mark size=2.5pt,color=black] at (belowx2) {\pgfuseplotmark{*}};

     \node [right=1.8cm of belowxV] (rho){};
     \node [above=0.0cm of rho] (rho1){$\rho_{RAB}$};

     \newcommand\Square[1]{+(-#1,-#1) rectangle +(#1,#1)}
     \draw (belowxV) \Square{20pt} ; 
     \draw (belowx1) \Square{20pt} ;
     \draw (belowx2) \Square{20pt} ;
     %\draw (belowx3) \Square{15pt} ;

     \draw [->] (xV) -- (belowxVbefore);
     \draw [->] (x1) -- (belowx1before);
     \draw [->] (x2) -- (belowx2before);
     %\draw [->] (x3) -- (belowx3before);

     \draw [-, transform canvas={xshift=0pt, yshift = 0 pt}, quantum] (belowxV) -- (belowx2) {};

\end{tikzpicture}
\caption{Schematic representation of an extended non-local game. The undulated line represents the tripartite quantum state $\rho_{RAB}$ shared amongst the three parties.  }
\label{fig:game}
\end{figure}

\begin{definition} A \emph{monogamy-of-entanglement game} \cite{TomamichelMonogamyGame2013} is an extended non-local game where $x=y=z$, i.e.\ all the parties get the same question, and the referee performs the measurement $\{V^x_a\}$. The winning set is given by $W=\{(x,v,a,b)\mid v=a=b \}$. Then, a monogamy-of-entanglement game will be denoted by the tuple
\begin{equation}
    \mathcal G=(q, V).
\end{equation}
\end{definition}

To illustrate the concept of extended non-local games, consider the following example introduced in \cite{TomamichelMonogamyGame2013}.

\begin{ex} \label{ex:BB84MoEGame} The \emph{BB84 monogamy-of-entanglement game} is described as follows. Alice and Bob prepare a quantum state and give a qubit to the referee, who, uniformly at random, choses to measure it either in the computational or the Hadamard basis. Upon knowing the choice of basis, the task of Alice and Bob is to guess the measurement outcome. Using the above terminology, the game is given by
\begin{equation}
    \mathcal{G}^{BB84}=\left(q(x)=\frac{1}{\abs{\mathcal{X}}},V=\{V^x_a=H^x\ketbra{a}{a}H^x\}_{x,a}, W=\{(x,v,a,b)\mid v=a=b\}_{x,v,a,b}\right),
\end{equation}
with $\mathcal X=\mathcal V=\mathcal A=\mathcal B=\{0,1\}$.  
\end{ex}

The most general thing that Alice and Bob can do using quantum mechanics is to perform POVMs $\{A^x_a\}$ and $\{B^y_b\}$ on their local registers to answer $a$ and $b$, respectively. A \emph{quantum strategy} $S_Q$ is defined by the tuple $\{\rho_{RAB},A^x_a,B^y_b\}_{x,y,a,b}$, and the quantum winning probability of the extended non-local game $\mathcal{G}$ is defined as 
\begin{equation}\label{eq:w_Q}
    \omega_{Q}(\mathcal{G}):=\sup\sum_{(z,x,y,v,a,b)\in W}q(z,x,y)\tr{(V^{z}_{v}\otimes A^x_a\otimes B^y_b) \rho_{RAB}},
\end{equation}
where the supremum (sup) is taken over all quantum states $\rho_{RAB}$, all POVMs $\{A^x_a\}$ and $\{B^y_b\}$ over all possible dimensions of $\mathcal{H}_A$ and $\mathcal{H}_B$. All the supremums in this work will be taken over the same conditions. A particular case of quantum strategies that will appear in this work are the \emph{unentangled} strategies. These correspond to quantum strategies $S_Q$ for which Alice and Bob prepare an unentangled initial state, i.e.\ of the form ${\rho_{RAB}=\sum_\lambda p_\lambda\rho_R^\lambda\otimes\rho_A^\lambda\otimes\rho_B^\lambda,}$ where $p_\lambda\geq 0$ are such that $\sum_\lambda p_\lambda=1$. Notice that the winning probability of unentangled strategies can be attained by Alice and Bob acting only classically (using shared randomness), since the following holds:
\begin{equation}
\begin{split}\label{eq:shared_randomness}
    \tr{(V^{z}_{v}\otimes A^x_a\otimes B^y_b) \rho_{RAB}}=\sum_\lambda p_\lambda\tr{V^{z}_{v}\rho_R^\lambda}\tr{ A^x_a\rho_A^\lambda}\tr{B^y_b\rho_B^\lambda}.
\end{split}
\end{equation}
Therefore, we will refer to unentangled strategies as those where Alice and Bob send a quantum state to the referee but only perform local operations based on shared classical randomness, reducing to $\rho_{RAB}=\rho_R$. Note that, since in this case the winning probability will be given by a convex combination, the optimal values will be obtained by extremal points and thus, by the proper choice of a pure state $\ketbra{\phi}{\phi}_R$, we have that 
\begin{equation*}
\begin{split}
     \sum_{(z,x,y,v,a,b)\in W}\!\!\!\!\!\!\!\!\!\!q(z,x,y)\sum_\lambda p_\lambda\tr{V^{z}_{v}\rho_R^\lambda}\tr{ A^x_a\rho_A^\lambda}\tr{B^y_b\rho_B^\lambda}\leq\!\!\!\!\!\!\!\sum_{(z,x,y,v,a,b)\in W}\!\!\!\!\!\!\!\!\!\!q(z,x,y)\tr{V^z_v\ketbra{\phi}{\phi}}p(a|x)p(b|y),
\end{split}
\end{equation*}
where $p(a|x)$ and $p(b|y)$ are the (classical) probabilities to output $a$ and $b$, given $x$ and $y$, respectively. For $\mathcal{G}^{BB84}$, it was shown~\cite{TomamichelMonogamyGame2013} that $\omega_Q(\mathcal{G}^{BB84})=\cos^2\left(\frac{\pi}{8}\right)$, which is attained by the unentangled strategy $S_{Q}^{BB84}=\{\ketbra{\phi}{\phi},A_a^x=\delta_{a0},B_a^x=\delta_{a0}\}$, where $\ket{\phi}=\cos\frac{\pi}{8}\ket{0}+\sin\frac{\pi}{8}\ket{1}$.

A broader class of strategies, encompassing quantum strategies as a special case, is given by the so-called \emph{commuting strategies} which consist of POVMs $\{A^x_a\}$ and $\{B^y_b\}$ on Alice's and Bob's joint system $\mathcal H$ such that $[A^x_a,B^y_b]=0$, i.e.~they commute, and a joint tripartite state $\rho_{RAB}$. The respective commuting winning probability of a non-local game $\mathcal{G}$ is defined as 
\begin{equation}\label{eq:w_comm}
    \omega_{comm}(\mathcal{G}):=\sup\sum_{(z,x,y,v,a,b)\in W}q(z,x,y)\tr{(V^{z}_{v}\otimes A^x_a B^y_b) \rho_{RAB}}.
\end{equation}

By extending Alice's and Bob's Hilbert spaces, if necessary, $\rho_{RAB}$ can be taken as a pure state, and $\{A^x_a\}$ and $\{B^y_b\}$ projective measurements. Since the quantum strategies are a subset of the commuting strategies, we have that, for all extended non-local games $\mathcal G$,
\begin{equation}
     \omega_{Q}(\mathcal{G})\leq  \omega_{comm}(\mathcal{G}).
\end{equation}

A particularly interesting case arises when an extended non-local game is sequentially played many times so that one can extract statistics from the answers of the players. 
Such a case arises, for example, when monogamy-of-entanglement games are used to analyze security of certain quantum protocols \cite{TomamichelMonogamyGame2013}, where the two collaborative parties play the role of attackers/dishonest implementers of the protocol who want to emulate the behavior of an honest party towards the referee. 

In a realistic implementation of those protocols, the honest party will not perform them with perfect correctness, and one might expect a certain distribution over the possible outcomes. For example, there will be answers that will be wrong with a certain probability due to e.g.\ experimental measurement errors or noisy quantum channels. However, we can find other types of answers that are wrong and that are simply inconsistent with the behavior of an honest executor of the protocol, for instance if she broadcasts classical information and this information is received differently in different locations. Moreover, in communication tasks, a sizable fraction of qubits get lost and, therefore, we expect a certain ratio of `empty' answers ($\perp$). Thus, it is natural to ask whether Alice and Bob's winning probability of a certain extended non-local game decreases if they not only have to answer correctly, but they have to emulate  the distribution of the other (different) incorrect answers and the transmission rate. In order to analyze these games where Alice and Bob are expected to answer certain answers with a certain probability, we introduce the concept of \emph{lossy-and-constrained} extended non-local games.  To clarify these ideas, we now present an example before introducing the formal definitions of constrained and lossy extended non-local games.

\begin{ex} Consider the $\mathcal G^{BB84}$ game, described in Example~\ref{ex:BB84MoEGame}.
\begin{enumerate}
    \item \emph{Constrained version (see Section~\ref{section:cons_BB84} for more details):} Motivated by security of quantum cryptographic protocols, this game is used to reduce their security where Alice and Bob pretend to mimic the action of an honest prover in a given protocol, see Section~\ref{sec:applications} for more details. In such a case, Alice and Bob have to coordinate their actions, and a natural imposition is given by forbidding them to answer differently. This is done by imposing that for every  strategy $S_{Q}=\{\rho_{RAB},A^x_a,B^y_b\}_{x,y,a,b}$,
\begin{equation}
    \sum_{x,a,a'\neq b'}q(x)\tr{(V^{x}_{a}\otimes A^x_{a'} \otimes B^x_{b'}) \rho_{RAB}}=0.
\end{equation}
\item \emph{Lossy version  (see Section~\ref{section:lossy_moe_games} for more details):} Motivated by losses in quantum communication, assume that Alice and Bob have the option to answer "inconclusive answer" ($\perp$) with probability $1-\eta$. This can be viewed as Alice and Bob claiming at will that the quantum messages have been lost, even if that was not the case. In this scenario, if they have to coordinately mimic a response rate $\eta$, this can be done by imposing that for every  strategy $S_{Q}=\{\rho_{RAB},A^x_a,B^y_b\}_{x,y,a,b}$,
\begin{equation} 
        \sum_{x}q(x)\tr{(\mathbb{I}_R\otimes A^x_{\perp}\otimes B^y_{\perp}) \rho_{RAB}}=1-\eta. 
    \end{equation}
\end{enumerate}

\end{ex}

\noindent We now formalize these concepts with the following definitions.

\begin{definition} Let $L\in\mathbb N$, for all $\ell \in\{1,\dots,L\}$, let $C=\{\alpha_{\ell}(v,a,b|z,x,y),c_{\ell}^{0},c_{\ell}^1\}_{(\ell,z,x,y,v,a,b)}$ (set of constraints) for $\alpha_{\ell}:\mathcal Z\times\mathcal{X}\times\mathcal{Y}\times\mathcal V\times\mathcal{A}\times\mathcal{B}\rightarrow \mathbb R$ and let $c_{\ell}^{0},c_{\ell}^1\in\mathbb R$ be such that $c_{\ell}^{0}\leq c_{\ell}^1$. 
We say that an extended non-local game is \emph{constrained}, and we denote it by $\mathcal{G}_{C}$, if for every strategy $S_{Q}=\{\rho_{RAB},A^x_a,B^y_b\}_{x,y,a,b}$,
\begin{equation}
    c_{\ell}^{0}\leq\sum_{z,x,y,v,a,b}\alpha_{\ell}(v,a,b | z,x,y)\tr{(V^{z}_{v}\otimes A^x_{a} \otimes B^y_{b}) \rho_{RAB}}\leq c_{\ell}^{1} \hspace{3mm}\forall\ell\in\{1,\dots,L\}.
\end{equation}
\end{definition}
\noindent See Section \ref{section:concrete_games} for examples of how to encode other problems as such a game.

\begin{definition}
    We say that an extended non-local game $\mathcal{G}$ is \emph{lossy} with parameter $\eta\in[0,1]$, and we denote it by $\mathcal{G}_{\eta}$, if the players are allowed to respond $a\in\mathcal{A}\cup\{\perp\}$ and $b\in\mathcal{B}\cup\{\perp\}$ in such a way that for every strategy  $S_{Q}=\{\rho_{RAB},A^x_a,B^y_b\}_{x,y,a,b}$, for $a\in\mathcal{A}\cup\{\perp\}$ and $b\in\mathcal{B}\cup\{\perp\}$,
    \begin{equation} \label{eq:lossy}
        \sum_{x,y}q(x,y)\tr{(\mathbb{I}_R\otimes A^x_{\perp}\otimes B^y_{\perp}) \rho_{RAB}}=1-\eta. 
    \end{equation}
\end{definition}

See more examples in Section~\ref{section:concrete_games}. We define a lossy extended non-local game by imposing that the average probability of the answers `$\perp$' is $1-\eta$. We stress that depending on the case to be considered, one can vary the definition and (i) impose that the probability for every input $x,y$ is $1-\eta$, i.e., $\tr{(\mathbb{I}_R\otimes A^x_{\perp} B^y_{\perp}) \rho_{RAB}}=1-\eta$ $\forall x,y\in \mathcal{X}\times\mathcal Y$ instead of on average, or (ii) impose different loss rates for Alice and Bob, which includes the case where one party answers $\perp$, while the other party does not. Although we stick to the above definition, the results presented in this work still apply if these alternative definitions are considered. 

\begin{remark} \label{rem:guess}\emph{(The guessing strategy)}. Consider a monogamy-of-entanglement game where the referee performs projective measurements $\{V^x_v\}$. If the distribution of the questions $x\in\mathcal X$ is uniform,  Alice and Bob can make a guess $\Tilde{x}$ for $x$ that will be correct with probability $\frac{1}{\abs{\mathcal X}}$. They can pick a fixed answer $a\neq \perp$ and send a state $\ket\psi\in\mathcal{H}_R$ that lives in the subspace that $V^{\Tilde{x}}_a$ projects onto, i.e.~$V^{\Tilde{x}}_a\ket{\psi}=\ket{\psi}$. Then, if they receive $x=\Tilde{x}$, they answer $a$ and otherwise, they answer $\perp$. This strategy based on guessing $x$ is such that Alice and Bob are going to be correct with probability~$\frac{1}{\abs{\mathcal X}}$ and will answer $\perp$ with probability~$1-\frac{1}{\abs{\mathcal X}}$; they will never give a wrong answer. 

If the distribution is not uniform, then Alice and Bob guess $\Tilde{x}$ with the highest probability, which will be at least $\frac{1}{\abs{\mathcal X}}$, and then they will be correct with probability at least $\frac{1}{\abs{\mathcal X}}$. 
    
\end{remark}

Because of the general strategy described in Remark~\ref{rem:guess}, throughout this paper we will consider the range $\eta\geq\frac{1}{\abs{\mathcal X}}$, unless explicitly stated otherwise.

If an extended non-local game $\mathcal G$ is both lossy and constrained, we say that it is  \emph{lossy-and-constrained}, and we denote it by $\mathcal{G}_{C,\eta}$. The quantum and commuting winning probabilities of $\mathcal{G}_{C,\eta}$ will be denoted by $\omega_{Q}(\mathcal{G}_{C,\eta})$ and $\omega_{comm}(\mathcal{G}_{C,\eta})$, respectively.

\subsection{Convergence of $\omega_{comm}(\mathcal{G}_{C,\eta})$ via SDPs}\label{section:convergence_proof}

In \cite{Extended_non-local_games_andMoE_games}, it was shown that similar ideas as in \cite{NPA2008} could be applied to construct a hierarchy of SDPs that upper bound the average winning probability of a given extended non-local game, which converges to its optimal value. Here, we show that when slightly modifying these games, by considering their lossy-and-constrained versions, the results still apply. For completeness, in this section we reproduce the proof in \cite{Extended_non-local_games_andMoE_games} to show that it applies in our case whenever the games that are considered are modified by their constraints and the loss.

Following Section 3 in \cite{Extended_non-local_games_andMoE_games}, consider the set \begin{equation}
    \Sigma=(\mathcal{X}\times\mathcal{A})\sqcup(\mathcal{Y}\times\mathcal{B}),
\end{equation}
where $\sqcup$ denoted the disjoint union. The set of strings of length at most $k$ over the alphabet $\Sigma$ will be denoted by $\Sigma^{\leq k}$, the set of all strings of finite length over $\Sigma$ will be denoted by $\Sigma^*$. Finally, the reverse of a string $s$ will be denoted by $s^R$, and the empty string will be denoted by $\varepsilon$. 
For example, 
\begin{equation}
\begin{split}
    \Sigma^{\leq 0}&=\{\varepsilon\},\\
    \Sigma^{\leq 1}&=\Sigma^{\leq0}\cup\{(x,a),(y,b)\}_{(x,a,y,b)\in \mathcal{X}\times\mathcal{A}\times\mathcal{Y}\times\mathcal{B}},\\
    \Sigma^{\leq 2}&= \Sigma^{\leq 1}\cup\{(x,a,x',a'),(y,b,y',b'),(x,a,y,b),(y,b,x,a)\}_{(x,x',a,a',y,y',b,b')\in \mathcal{X}^{\times2}\times\mathcal{A}^{\times2}\times\mathcal{Y}^{\times2}\times\mathcal{B}^{\times2}}.
\end{split}
\end{equation}
For later purposes, we will consider the set defined for $k=`1+AB$', which is given by 
\begin{equation}
    \Sigma^{\leq {`1+AB\text{'}}}= \Sigma^{\leq 1}\cup\{(x,a,y,b)\}_{(x,a,y,b)\in \mathcal{X}\times\mathcal{A}\times\mathcal{Y}\times\mathcal{B}}\subset\Sigma^{\leq 2}.
\end{equation}
Consider the equivalence relation  $\sim$ on $\Sigma^*$ defined by:
\begin{enumerate}
    \item For every $s,t\in\Sigma^*$, and $\sigma$ in $\Sigma$, $s\sigma t\sim s \sigma \sigma t$, 
    \item For every $s,t\in\Sigma^*$, $\sigma_A\in\mathcal{X}\times\mathcal{A}$, and  $\sigma_B\in\mathcal{Y}\times\mathcal{B}$, $s\sigma_A\sigma_B t\sim s\sigma_B\sigma_A t$.
\end{enumerate}

\begin{definition} \emph{(Admissible function \cite{Extended_non-local_games_andMoE_games})}. A function $\phi:\Sigma^*\rightarrow\mathbb C$ is said to be \emph{admissible} if and only if the following conditions are satisfied:
\begin{enumerate}
    \item For every $s,t\in\Sigma^*$, and every $(x,y)\in\mathcal{X}\times\mathcal{Y}$, 
    \begin{equation} \label{eq:admissible1}
        \sum_{a\in\mathcal{A}}\phi(s(x,a)t)=\phi(s,t)= \sum_{b\in\mathcal{B}}\phi(s(y,b)t).
    \end{equation}
    \item For every $s,t\in\Sigma^*$, every $(x,y)\in\mathcal{X}\times\mathcal{Y}$, and every $a\neq a'\in\mathcal{A}$, $b\neq b'\in\mathcal{B}$,
    \begin{equation}\label{eq:admissible2}
        \phi(s(x,a)(x,a')t)=0=\phi(s(y,b)(y,b')t).
    \end{equation}
    \item For every $s,t\in\Sigma^*$ such that $s\sim t$, 
    \begin{equation}\label{eq:admissible3}
        \phi(s)=\phi(t).
    \end{equation}
A function $\phi:\Sigma^{\leq k}\rightarrow\mathbb C$ is said to be admissible if and only if \eqref{eq:admissible1}-\eqref{eq:admissible3} hold, given that the arguments' length is such that $\phi$ is defined. 
\end{enumerate}
    
\end{definition}
We extend the definition of $k$th order admissible matrices in \cite{Extended_non-local_games_andMoE_games}. 

\begin{definition}
    For every $1\leq k\in\mathbb Z$, let $M^{(k)}$ be a block matrix of the form 
\begin{equation}\label{eq:M^k}
    M^{(k)} = \begin{pmatrix}
        M_{1,1}^{(k)} & \cdots & M_{1,m}^{(k)} \\
        \vdots & \ddots & \vdots \\
        M_{m,1}^{(k)} & \cdots & M_{m,m}^{(k)}
    \end{pmatrix}
\end{equation}
where $ M_{i,j}^{(k)}:\Sigma^{\leq k}\times\Sigma^{\leq k}\rightarrow\mathbb C$ for every $i,j\in\{1,\ldots,m\}$. The matrix $M^{k}(s,t)$ is defined as the matrix with entries $M^{(k)}_{i,j}(s,t)$. 
We say that $M^{(k)}$ is a lossy-and-extended $k$th order admissible matrix for $\mathcal{G}_{C,\eta}$ with constraints $C=\{\alpha_{\ell}(v,a,b|z,x,y),c_{\ell}^{0},c_{\ell}^1\}_{(\ell,z,x,y,v,a,b)}$, if the following conditions are satisfied:
\begin{enumerate}
    \item For every $i,j\in\{1,\ldots,m\}$, there exists an admissible function $\phi_{i,j}:\Sigma^{2k}\rightarrow\mathbb C$ such that for every $s,t\in\Sigma^{\leq k}$, 
    \begin{equation} \label{eq:lin_cons1}
        M^{(k)}_{i,j}(s,t)=\phi_{i,j}(s^Rt).
    \end{equation}
    \item The following holds
    \begin{equation}\label{eq:trM=1}
        \tr{M^{(k)}(\varepsilon,\varepsilon)}=1.
    \end{equation}
    \item For all $\ell\in\{1,\dots,L\}$, 
    \begin{equation} \label{eq:tr_cons}
        c_{\ell}^{0}\leq\sum_{z,x,y,v,a,b}\alpha_{\ell}(v,a,b|z,x,y)\tr{V^{z}_{v}M^{(k)}((x,a),(y,b))}\leq c_{\ell}^{1},
    \end{equation}
    and \begin{equation}
        \sum_{x,y}q(x,y)\tr{M^{(k)}((x,\perp),(y,\perp))}=1-\eta. \label{eq:tr_eta}
    \end{equation}
    \item The matrix $M^{(k)}$ is positive semidefinite. 
\end{enumerate}
\end{definition}

The definition of a $k$th order admissible matrix defined in \cite{Extended_non-local_games_andMoE_games} is the same, except without condition 3. 

For a lossy-and-extended $k$th order admissible matrix for $\mathcal{G}_{C,\eta}$ with constraints\\ $C=\{\alpha_{\ell}(z,x,y|v,a,b),c_{\ell}^{0},c_{\ell}^1\}_{(\ell,z,x,y,v,a,b)}$, consider 
\begin{equation}
    \omega_{admiss}^{(k)}(\mathcal{\mathcal{G}_{C,\eta}}):=\sup\sum_{(z,x,y,v,a,b)\in W}q(z,x,y)\tr{V^{z}_{v}M^{(k)}((x,a),(y,b))}.
\end{equation}
Here the supremum is taken over all the strategies that fulfill the constraints given by $C$ and are consistent with $\eta$ in the sense of \eqref{eq:lossy}. 
\begin{remark}
    The value $\omega_{admiss}^{(k)}(\mathcal{\mathcal{G}_{C,\eta}})$ can be computed by the following SDP:
    \begin{equation}\label{eq:SDP}
    \begin{split}
         &\max \sum_{(z,x,y,v,a,b)\in W}q(z,x,y)\tr{V^{z}_{v}M^{(k)}((x,a),(y,b))}\\
         &\text{subject to the linear constraints given by \eqref{eq:lin_cons1}-\eqref{eq:tr_eta},}\\
        & \hspace{0cm} \text{ and } M^{(k)}\succeq0. 
    \end{split}
    \end{equation}

\end{remark}
Recall that semidefinite programs (SDPs) can be solved in polynomial time to any fixed prescribed precision, see e.g.\ \cite{LAURENT2005393,Boyd_Vandenberghe_2004}.

\begin{lemma}\label{lemma:upper_bound} For all $ k\geq  k'\geq 1$,
   \begin{equation}
   \omega_{admiss}^{(k)}(\mathcal{\mathcal{G}_{C,\eta}}) \geq \omega_{admiss}^{(k')}(\mathcal{\mathcal{G}_{C,\eta}})\geq \omega_{comm}(\mathcal{\mathcal{G}_{C,\eta}}).
\end{equation} 
\end{lemma}

\begin{proof} The function $K$ defined as $K(a,b|x,y)=\ptr{AB}{(\mathbb I_{R}\otimes A^x_a B^y_b)\rho_{RAB}}$ is said to be a \emph{commuting measurement assemblage}. Then,  \eqref{eq:w_Q} can be rewritten as  \begin{equation}      \omega_{comm}(\mathcal{G})=\sup\sum_{(z,x,y,v,a,b)\in W}q(z,x,y)\tr{V^{z}_{v}K(a,b|x,y)}. \end{equation}
We will show that given a commuting measurement assemblage $K$, for every integer $k\geq 1$, there exists a lossy-and-extended $k$th order admissible matrix for $\mathcal{G}_{C,\eta}$ with constraints $C$. 

Consider a commuting strategy $S_{comm}=\{\ket{\psi},A^x_a,B^y_b\}$ for $\mathcal{G}_{C,\eta}$, with $\ket{\psi}\in\mathcal H_R\otimes \mathcal{H}_A\otimes \mathcal{H}_B$  and $A^x_a,B^y_b$ projective measurements (recall that a strategy can be taken as a pure state and projective measurements). Consider a Schmidt decomposition of $\ket\psi$ given by $\ket\psi=\sum_{i=1}^m\lambda_i\ket{i}_{R}\ket{\psi_i}_{AB}$, where $\{\ket{i}\}$ is the standard basis of $\mathcal{H}_R$ and $\ket{\psi_i}_{AB}\in\mathcal{H}_A\otimes\mathcal{H}_B$. Let $\ket{\Tilde{\psi}_i}:=\lambda_i\ket{\psi_i}$.

Define 
\begin{equation}
    \Pi^z_c:=\begin{cases}
        A^z_c \text{ if } (c,z)\in\mathcal{A}\times\mathcal{X},\\
        B^z_c \text{ if } (c,z)\in\mathcal{B}\times\mathcal{Y}.
        
    \end{cases}
\end{equation}
Consider $M^{(k)}$ defined by 
\begin{equation}
    M_{i,j}^{(k)}(s,t)=\phi_{i,j}(s^Rt) \text{ } \forall i,j\in\{1,\ldots, m\},
\end{equation}
where $\phi_{i,j}$ are defined as 
\begin{equation}
    \phi_{i,j}((z_1,c_1)\dots(z_l,c_l))=\bra{\Tilde\psi_j}\Pi^{z_1}_{c_1}\cdots\Pi^{z_l}_{c_l}\ket{\Tilde\psi_i}.
\end{equation}
The functions $\phi_{i,j}$ fulfill \eqref{eq:admissible1}-\eqref{eq:admissible3}, and thus they are admissible functions. 
Notice that
\begin{equation}
\begin{split}
    1&=\tr{\ketbra{\psi}{\psi}}=\tr{\sum_{i,j=1}^m\ket{i}\ket{\Tilde\psi_i}\bra{\Tilde\psi_j}\bra{j}}=
    \sum_{i,j=1}^m\braket{i}{j}\braket{\Tilde\psi_i}{\Tilde\psi_j}\\&=\sum_{i=1}^m\braket{\Tilde\psi_i}{\Tilde\psi_i}
    =\sum_{i=1}^m\phi_{i,i}(\varepsilon,\varepsilon)=\sum_{i=1}^mM_{i,i}^{(k)}(\varepsilon,\varepsilon)=\tr{M^{k}(\varepsilon,\varepsilon)},
    \end{split}
\end{equation}
and therefore \eqref{eq:trM=1} is fulfilled. Moreover, $S_{comm}$ is such that equations \eqref{eq:tr_cons} and \eqref{eq:tr_eta} hold. Finally, since $M^{(k)}$ is a Gram matrix by construction, it is positive semidefinite. Thus, $M^{(k)}$ is a lossy-and-extended $k$th order admissible matrix for $\mathcal{G}_{C,\eta}$ with constraints $C$. 

It remains to see that $M^{(k)}$ corresponds to $K$. Since
\begin{equation}
\begin{split}
    K&(a,b|x,y)=\ptr{AB}{(\mathbb I_R\otimes A^x_aB^y_b)\sum_{i,j=1}^{m}\lambda_i\lambda_j\ket{j}\ket{\psi_j}\bra{\psi_i}\bra{i})}\\&=\sum_{k=1}^{\dim\mathcal{H}}\bra{\psi_k}(\mathbb I_R\otimes A^x_aB^y_b)\sum_{i,j=1}^{m}\lambda_i\lambda_j\ket{j}\ket{\psi_j}\bra{\psi_i}\bra{i}) \ket{\psi_k}=\sum_{i,j=1}^{m}\ketbra{i}{j}\lambda_j\bra{\psi_j}A^x_aB^y_b\ket{\psi_i}\lambda_i\\&=\sum_{i,j=1}^{m}\ketbra{i}{j}\bra{\Tilde\psi_j}A^x_aB^y_b\ket{\Tilde\psi_i}=\sum_{i,j=1}^{m}\ketbra{i}{j}M^{(k)}_{i,j}((x,a),(y,b)),
    \end{split}
\end{equation}
we have that $K(a,b|x,y)=M^{(k)}((x,a),(y,b))$. 

Finally, the inequality $\omega_{admiss}^{(k)}(\mathcal{\mathcal{G}_{C,\eta}}) \geq \omega_{admiss}^{(k')}$ for $k\geq k'$  comes from the fact that both are obtained with an SDP with the same objective function but the latter is obtained with more constraints.
\end{proof}

\begin{lemma} \label{lemma:M<=1}(Lemma 5.2. in \cite{russo2017extended}). Let $M^{(k)}$ be as in \eqref{eq:M^k}. Then, for all $i,j\in\{1,\ldots,m\}, s,t\in\Sigma^{\leq k}$,
\begin{equation}
    \abs{M_{i,j}^{(k)}(s,t)}\leq 1.
\end{equation}
\end{lemma}

\begin{lemma}\label{lem:banach_alaoglu} \emph{(Banach--Alaoglu theorem for separable spaces \cite{Banach_Alaoglu_thm})}. Let $X$ be a separable normed space. Then, the closed unit ball in $X^*$ is sequentially weak*-compact. That is, for every sequence $\{\mu_n\}_{n\in\mathbb N}$ in $X^*$ with $\|\mu_n\|\leq 1$ $\forall n\in\mathbb N$, there exists a subsequence $\{\mu_{n_k}\}_{k\in\mathbb N}$ such that weak* converges to $\mu\in X^*$: $\lim_{k\rightarrow\infty}\mu_{n_k}=\mu$.
\end{lemma}

\begin{theorem}\label{thm:convergence}The following holds
\begin{equation}
        \lim_{k\rightarrow\infty}\omega_{admiss}^{(k)}(\mathcal{\mathcal{G}_{C,\eta}}) =\omega_{comm}(\mathcal{\mathcal{G}_{C,\eta}}) .
\end{equation}
\end{theorem}

In summary, for a given game $\mathcal{G}_{\eta,C}$, Lemma \ref{lemma:upper_bound} shows us that computing the SDP \eqref{eq:SDP} will provide an upper bound to the value $\omega_{Q}(\mathcal{G}_{\eta,C})$. Moreover, Theorem \ref{thm:convergence} tells us that the upper bounds will actually converge to the optimal value of $\omega_{comm}(\mathcal{G}_{\eta,C})$. 

\begin{proof}
From Lemma \ref{lemma:upper_bound} we have a non-increasing sequence $\{\omega_{admiss}^{(k)}(\mathcal{\mathcal{G}_{C,\eta}})\}$ that for every $k$, $\omega_{admiss}^{(k)}(\mathcal{\mathcal{G}_{C,\eta}}) \geq \omega_{comm}(\mathcal{\mathcal{G}_{C,\eta}}).$
We will first show the convergence of $\omega_{admiss}^{(k)}(\mathcal{\mathcal{G}_{C,\eta}})$, and finally that it converges to $\omega_{comm}(\mathcal{\mathcal{G}_{C,\eta}})$.

By Lemma \ref{lemma:M<=1}, we have that $\abs{M_{i,j}^{(k)}(s,t)}\leq 1$ $\forall i,j\in\{1,\ldots,m\}, s,t\in\Sigma^{\leq k}$. By the Banach-Alaoglu theorem for separable spaces (Lemma \ref{lem:banach_alaoglu}), there exists a subsequence $\{M_{i,j}^{(k_l)}(s,t)\}_{l\in \mathbb N}$ and $M_{i,j}(s,t)$ such that $\lim_{l\rightarrow\infty}=M_{i,j}(s,t)$. 
Let 
\begin{equation}
    M = \begin{pmatrix}
        M_{1,1} & \cdots & M_{1,m} \\
        \vdots & \ddots & \vdots \\
        M_{m,1} & \cdots & M_{m,m}
    \end{pmatrix},
\end{equation}
then, by construction $M^{(k_l)}\rightarrow M$.
The $M_{i,j}$ are such that
\begin{enumerate}
    \item For every $i,j\in\{1,\ldots,m\}$, there exists an admissible function $\phi_{i,j}:\Sigma^{*}\rightarrow\mathbb C$ such that for every $s,t\in\Sigma^{\leq k}$, 
    \begin{equation}
        M^{(k)}_{i,j}(s,t)=\phi_{i,j}(s^Rt).
    \end{equation}
    \item Equation \eqref{eq:trM=1} holds,
    \item For all $\ell\in\{1,\dots,L\}$, equations \eqref{eq:tr_cons}, and \eqref{eq:tr_eta} hold.
    \item The matrix $M$ is positive semidefinite. 
\end{enumerate}

We want to see that $M$ defines a commuting measurement strategy $S_{comm}=\{\ket\psi,A^x_a,B^y_b\}$. Since $M$ is positive semidefinite, there exists vectors $\ket{\psi_{i,s}}\in\mathcal{H}$ for $i\in\{1,\ldots,m\}$  and $s\in\Sigma^*$, for a separable Hilbert space $\mathcal H$,  such that 
\begin{equation}\label{eq:M_ij=braket}
    M_{i,j}(s,t)=\braket{\psi_{j,s}}{\psi_{i,t}}.
\end{equation}
 
Assume $\mathcal{H}$ is spanned by $\{\ket{\psi_{i,s}}\}$, otherwise let $\mathcal{H}=\text{span} \{\ket{\psi_{i,s}}\}$. 
\begin{itemize}
    \item Consider the state 
    \begin{equation}
        \ket\psi=\sum_{j=1}^m\ket j\ket{\psi_{j,\varepsilon}}\in \mathcal{H}_R\otimes\mathcal{H}, 
    \end{equation}
    which is normalized, since
    \begin{equation}
        \|\ket{\psi}\|^2=\braket{\psi}{\psi}=\sum_{j=1}^m\braket{j}{j}\braket{\psi_{j,\varepsilon}}{\psi_{j,\varepsilon}}=\sum_{j=1}^m\braket{\psi_{j,\varepsilon}}{\psi_{j,\varepsilon}}=\sum_{j=1}^mM_{j,j}(\varepsilon,\varepsilon)=1,
    \end{equation}
    where we used \eqref{eq:M_ij=braket} and \eqref{eq:trM=1}. 

    \item For all $(z,c)\in\Sigma$, let $\Pi^z_c$ be the projection operator onto the  following space
    \begin{equation}\label{eq:spanPi}
        \text{span}\{\ket{\psi_{j,(z,c)s}}:j\in\{1,\ldots,m\},s\in\Sigma^* \}.
    \end{equation}
    We want to see that these operators are projective measurements, and that the ones corresponding to Alice commute with the ones corresponding to Bob. 
    Notice that \eqref{eq:spanPi} is the image of $\Pi^z_c$, which we denote by Im$(\Pi^z_c)$. 
    For all $j\in\{1,\ldots,m\}$ and $s\in\Sigma^*$,
    \begin{equation}
    \begin{split}
        \braket{\psi_{j,(z,c)t}}{\psi_{j,s}}&=M_{i,j}((z,c)t,s)=\phi_{i,j}(t^R(z,c)s)=\phi_{i,j}(t^R(z,c)(z,c)s)\\&=M_{i,j}((z,c)t,(z,c)s)=\braket{\psi_{j,(z,c)t}}{\psi_{i,(z,c)s}}.
    \end{split}
    \end{equation}
By linearity, for every $\ket{\varphi_1}\in \Im(\Pi^z_c)$,
\begin{equation}\label{eq:psi_js}
    \braket{\varphi_1}{\psi_{j,s}}=\braket{\varphi_1}{\psi_{j,(z,c)s}}.
\end{equation}
Let $(\Pi^z_c)^\perp$ denote the projection onto $\Im(\Pi^z_c)^\perp$ (orthogonal complement), so that  $\Pi^z_c+(\Pi^z_c)^\perp=\mathbb I$. Then,
\begin{equation}
\begin{split}
    \braket{\varphi_1}{\psi_{j,s}}=\braket{\varphi_1}{(\Pi^z_c+(\Pi^z_c)^\perp)\psi_{j,s}}
    =
   \bra{\varphi_1}\Pi^z_c\ket{\psi_{j,s}}+\bra{\varphi_1}(\Pi^z_c)^\perp\ket{\psi_{j,s}}= \bra{\varphi_1}\Pi^z_c\ket{\psi_{j,s}},
   \end{split}
\end{equation}
where we used that $(\Pi^z_c)^\perp\ket{\psi_{j,s}}\in \Im(\Pi^z_c)^\perp$ and therefore, $\bra{\varphi_1}(\Pi^z_c)^\perp\ket{\psi_{j,s}}=0$. Using \eqref{eq:psi_js}, 
we have that 
\begin{equation}
    \bra{\varphi_1}\Pi^z_c\ket{\psi_{j,s}}=\braket{\varphi_1}{\psi_{j,(z,c)s}}.
\end{equation}

Take an arbitrary $\ket{\varphi}\in\mathcal{H}$, then $\ket\varphi=\ket{\varphi_1}+\ket{\varphi_1^{\perp}}$, for $\ket{\varphi_1}\in \Im(\Pi^z_c)$ and $\ket{\varphi_1^{\perp}}\in \Im(\Pi^z_c)^\perp$. Then, 
\begin{equation}\label{eq:equal_in_ImPi}
\bra{\varphi}\Pi^z_c\ket{\psi_{j,s}}=\bra{\varphi_1}\Pi^z_c\ket{\psi_{j,s}}+\bra{\varphi_1^\perp}\Pi^z_c\ket{\psi_{j,s}}=\bra{\varphi_1}\Pi^z_c\ket{\psi_{j,s}},
\end{equation}
\begin{equation}
\braket{\varphi}{\psi_{j,(z,c)s}}=\braket{\varphi_1}{\psi_{j,(z,c)s}}+\braket{\varphi_1^\perp}{\psi_{j,(z,c)s}}=\braket{\varphi_1}{\psi_{j,(z,c)s}},
\end{equation}
where we used $\bra{\varphi_1^\perp}\Pi^z_c\ket{\psi_{j,s}}=0$ and $\braket{\varphi_1^\perp}{\psi_{j,(z,c)s}}=0$. Then, we have that, because of \eqref{eq:equal_in_ImPi}, for all $\ket\varphi\in\mathcal H$, and for all $\ket{\psi_{j,s}}$, $\bra{\varphi}\Pi^z_c\ket{\psi_{j,s}}=\braket{\varphi}{\psi_{j,(z,c)s}}$, and therefore, 
\begin{equation}\label{eq:Pipsi}
    \Pi^z_c\ket{\psi_{j,s}}=\ket{\psi_{j,(z,c)s}}.
\end{equation}
Then, for all $i,j\in\{1,\ldots,m\}$, $s,t\in\Sigma^*$, for $(z,c_1),(z,c_2)\in\Sigma$ with $c_1\neq c_2$,
\begin{equation}
    \bra{\psi_{j,t}}\Pi^z_{c_1}\Pi^z_{c_2}\ket{\psi_{i,s}}=\braket{\psi_{j,(z,c_1)t}}{\psi_{i,(z,c_2)s}}=M_{i,j}((z,c_1)t,(z,c_2)t)=\phi_{i,j}(t^R(z,c_1)(z,c_2)s)=0. 
\end{equation}
Therefore, we have that 
\begin{equation}
    \Pi^z_{c_1}\Pi^z_{c_2}=0. 
\end{equation}
Moreover, for all $i,j\in\{1,\ldots,m\}$, $s,t\in\Sigma^*,x\in\mathcal{X}$,
\begin{equation}
\sum_{a\in\mathcal{A}}\bra{\psi_{i,s}}\Pi^x_a\ket{\psi_{j,t}}=\sum_{a\in\mathcal{A}}\braket{\psi_{i,s}}{\psi_{j,(x,a)t}}=\sum_{a\in\mathcal A}\phi_{i,j}(s^R(x,a)t)=\phi_{i,j}(s^Rt)=\braket{\psi_{i,s}}{\psi_{j,t}},
\end{equation}
and thus, 
\begin{equation}
    \sum_{a\in\mathcal{A}}\Pi^x_a=\mathbb I. 
\end{equation}
Analogously, one finds that for all $y\in\mathcal Y$, 
\begin{equation}
    \sum_{b\in\mathcal{B}}\Pi^y_b=\mathbb I. 
\end{equation}

For all $i,j\in\{1,\ldots,m\},s,t\in\Sigma^*,x\in\mathcal{X},y\in\mathcal{Y},a\in\mathcal A, b\in\mathcal{B}$,
\begin{equation}
\begin{split}
     \bra{\psi_{i,s}}\Pi^x_a\Pi^y_b\ket{\psi_{j,t}}&=\braket{\psi_{i,(x,a)s}}{\psi_{j,(y,b)t}}=\phi_{i,j}(s^R(x,a)(y,b)t)=\phi_{i,j}(s^R(y,b)(x,a)t)\\&=\braket{\psi_{i,(y,b)s}}{\psi_{j,(x,a)t}}= \bra{\psi_{i,s}}\Pi^y_b\Pi^x_a\ket{\psi_{j,t}},
\end{split}
\end{equation}
and therefore
\begin{equation}
    [\Pi^x_a,\Pi^y_b]=0.
\end{equation}
\end{itemize}
Lastly, we see that, from \eqref{eq:Pipsi},
\begin{equation}
    M_{i,j}((x,a),(y,b))=\braket{\psi_{j,(x,a)}}{\psi_{i,(y,b)}}=\bra{\psi_{j,\varepsilon}}A^x_aB^y_b\ket{\psi_{i,\varepsilon}},
\end{equation}
and, thus, 
\begin{equation}
\begin{split}
     K(a,b\mid x,y)&=\ptr{\mathcal H}{(\mathbb I_R\otimes A^x_aB^y_b)\ketbra{\psi}{\psi}}
    =\sum_{l,k}\bra{\psi_{l,k}}(\mathbb I_R\otimes A^x_aB^y_b)\sum_{i,j}\ket{i}\ket{\psi_{i,\varepsilon}}\bra{\psi_{j,\varepsilon}}\bra{j}\ket{\psi_{l,k}}\\
    &=\sum_{i,j}\bra{\psi_{j,\varepsilon}}A^x_aB^y_b\ket{\psi_{i,\varepsilon}}\ketbra{i}{j}=\sum_{i,j}M_{i,j}((x,a),(y,b))\ketbra{i}{j}=M((x,a),(y,b)).
\end{split}
\end{equation}

\end{proof}

\section{Concrete games}\label{section:concrete_games}
In this section, we consider particular extended non-local games, and we analyze lossy and constrained versions of these games, inspired by experimental considerations computing the SDPs that give the values of $\omega_{admiss}^{(k)}$. We show that our results are consistent with previously known results, for some cases with a simpler proof, and we show new results, and tight results in Section~\ref{section:lossy_moe_games}. For a reader interested in applications in quantum cryptography, the games analyzed in this section will be used in Section~\ref{sec:applications} to prove security for cryptographic primitives.

\subsection{Constrained BB84 monogamy-of-entanglement game}\label{section:cons_BB84}

In~\cite{TomamichelMonogamyGame2013}, the following extended non-local game (the BB84 \emph{monogamy-of-entanglement game)} was introduced 
\begin{equation}
    \mathcal{G}^{BB84}=\left(q(x)=\frac{1}{\abs{\mathcal{X}}},V=\{V^x_a=H^x\ketbra{a}{a}H^x\}_{x,a}, W=\{(x,v,a,b)\mid v=a=b\}_{x,v,a,b}\right),
\end{equation}
with $\mathcal X=\mathcal V=\mathcal A=\mathcal B=\{0,1\}$. In this game, Alice and Bob prepare a quantum state and give a qubit to the referee, who measures it either in the computational or the Hadamard basis. The task of Alice and Bob is to guess the measurement outcome. In \cite{TomamichelMonogamyGame2013}, it was shown that $\omega_Q(\mathcal{G}^{BB84})=\cos^2\left(\frac{\pi}{8}\right)$, which is attained by the strategy $S_{Q}^{BB84}=\{\ketbra{\phi}{\phi},A_a^x=\delta_{a0},B_a^x=\delta_{a0}\}$, where $\ket{\phi}=\cos\frac{\pi}{8}\ket{0}+\sin\frac{\pi}{8}\ket{1}$.

We will consider the constrained version of $ \mathcal{G}^{BB84}$ given by imposing that Alice and Bob never answer differently, i.e.\ imposing that their probability of answering different is 0.
There are settings in quantum cryptography, for example in bit-commitment and quantum-position verification, where this constraint can be added naturally -- see Section~\ref{sec:rel_BC} and Section \ref{sec:QPV}.
In this case, the set of constraints $C$ is given by: for every  strategy $S_{Q}=\{\rho_{RAB},A^x_a,B^y_b\}_{x,y,a,b}$,
\begin{equation}
    \sum_{x,a,a'\neq b'}q(x)\tr{(V^{x}_{a}\otimes A^x_{a'} \otimes B^x_{b'}) \rho_{RAB}}=0.
\end{equation}
From Section \ref{section:convergence_proof}, we have that 
\begin{equation}
    \omega_Q(\mathcal{G}^{BB84}_C)\leq\omega_{comm}(\mathcal{G}^{BB84}_C)\leq \omega_{admiss}^{(k=1)}(\mathcal{G}^{BB84}_C).
\end{equation}
Solving the SDP \eqref{eq:SDP} to obtain the value of $\omega_{admiss}^{(k=1)}(\mathcal{G}^{BB84}_C)$, see code \cite{code}, we find 
\begin{equation}
    \omega_{admiss}^{(k=1)}(\mathcal{G}^{BB84}_C)=0.8535533905\simeq \cos^2\left(\frac{\pi}{8}\right),
\end{equation}
and therefore we see that constraining Alice and Bob by forbidding them to answer inconsistently does not lower their average winning probability. This is not surprising, since in \cite{TomamichelMonogamyGame2013} it was shown that the strategy $S_{Q}^{BB84}$, see above, gives the optimal value for $\mathcal{G}^{BB84}$. This strategy fulfills the constraints $C$, and therefore it is also optimal for $\mathcal{G}^{BB84}_C$.

Moreover, we see that for this game  giving inconsistent answers precludes attaining the optimal value. For instance, if we impose the constraints $C_\varepsilon$ given by 
\begin{equation}\label{eq:constraint_epsilon}
    \sum_{x,a,a'\neq b'}q(x)\tr{(V^{x}_{a}\otimes A^x_{a'} \otimes B^x_{b'}) \rho_{RAB}}\geq\varepsilon,
\end{equation}
for $\varepsilon\geq0$, thus imposing that the probability that Alice and Bob answer inconsistently is at least $\varepsilon$, the upper bounds obtained by $\omega_{admiss}^{(k=1)}(\mathcal{G}^{BB84}_{C_\varepsilon})$ for different values of $\varepsilon>0$ are lower than $\cos^2\left(\frac{\pi}{8}\right)$, see Fig.~\ref{Fig:QPV_epsilon} (black dots) for $0\leq\varepsilon\leq0.25$. This fact is consistent with the rigidity results for the $\mathcal{G}^{BB84}$ game shown in~\cite{RigidityMoE}.

\begin{figure}[h]
\centering
\includegraphics[width=125mm]{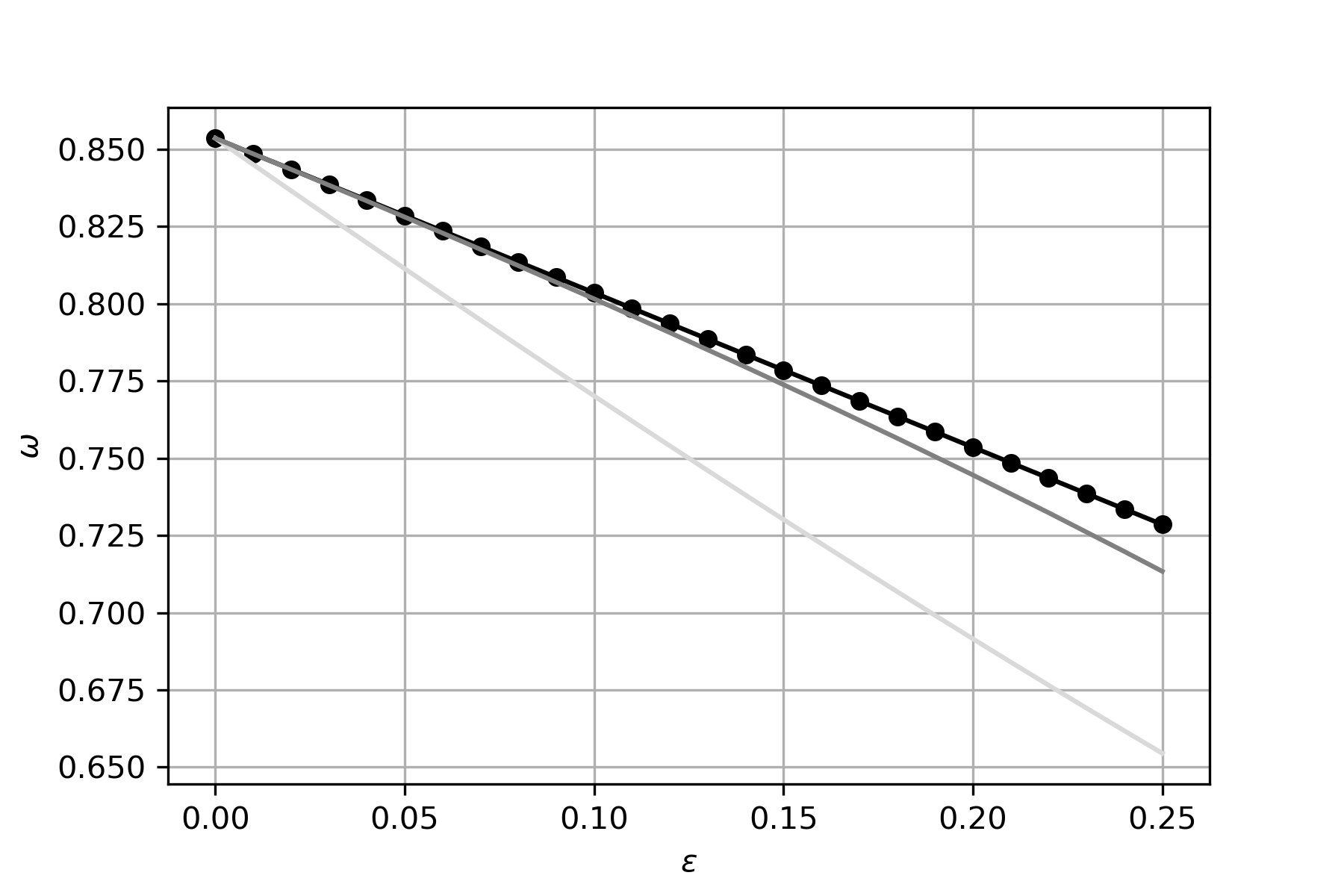}
\caption{Values of $\omega_{admiss}^{(k=1)}(\mathcal{G}^{BB84}_{C_\varepsilon})$ (black dots), optimal winning probability using unentangled strategies (light gray continuous line), 2-qubit state (dark gray continuous line) and 3-qubit state (black continuous line) for different values of $\varepsilon$. }
\label{Fig:QPV_epsilon}
\end{figure}

A numerical optimization over states $\ket{\phi}_{VAB}$ of dimension $2^3$ with real coefficients, where each register is 2-dimensional (3-qubit states), and 2-dimensional projective measurements for both Alice and Bob, shows that the results  obtained with the first level of the hierarchy ($\omega_{admiss}^{(k=1)}(\mathcal{G}^{BB84}_{C_\varepsilon})$) are \emph{tight}, up to numerical precision. See the continuous black line in Fig.~\ref{Fig:QPV_epsilon} for the values of the numerical optimization, which match the values of $\omega_{admiss}^{(k=1)}(\mathcal{G}^{BB84}_{C_\varepsilon})$ (black dots), and see \cite{code} for the explicit states $\ket{\phi}_{VAB}$ and the 2-dimensional projective measurements $\{A^x_a\}$ and $\{B^y_b\}$ for different values of $\varepsilon$ fulfilling the constraint~\eqref{eq:constraint_epsilon} and matching the upper bounds $\omega_{admiss}^{(k=1)}(\mathcal{G}^{BB84}_{C_\varepsilon})$.

In addition, the optimal values for $\varepsilon>0$, unlike for $\varepsilon=0$ (strategy $S^{BB84}_Q$), cannot be obtained by unentangled strategies, i.e.\ by just preparing a single qubit and sending it to the referee. A numerical optimization over unentangled strategies fulfilling the constraint \eqref{eq:constraint_epsilon}, see \cite{code}, shows that the optimal values are significantly lower than $\omega_{admiss}^{(k=1)}(\mathcal{G}^{BB84}_{C_\varepsilon})$, see light gray line in Fig.~\ref{Fig:QPV_epsilon}. Numerically optimizing over all states $\ket{\phi}_{VAB}$ of dimension $2^2$, with the referee's and Alice's registers being both 2-dimensional and Alice performing projective measurements, shows that the optimal value increases, being \emph{tight} (only) for small values of $\varepsilon$, see the dark gray line in Fig.~\ref{Fig:QPV_epsilon}.  Therefore, we see that imposing certain constraints forces Alice and Bob to utilize higher dimensional quantum systems to achieve the optimal winning probability. We would like to stress that we do not have an understanding of the fundamental reason behind this requirement. 

\begin{remark}
In order to analyze optimal strategies, if one is able to find an optimal strategy for an extended non-local game (or monogamy-of-entanglement game), one can constrain the game and see if certain restrictions on the strategies affect the optimal winning value. Finding a still-optimal perturbation of an optimal strategy in such way is a way to exclude possible rigidity results.
\end{remark}

\subsection{Alice guessing game with constraints}\label{section:Alice_guess_game}
We consider a slight variation of the $\mathcal G^{BB84}$ monogamy-of-entanglement game where only Alice has to guess the measurement outcome of the referee, described as follows:
\begin{equation}
    \mathcal G^{\text{A } n \text{ guess}}:=\left(q(x)=\frac{1}{\abs{\mathcal X}}, V=\{V^x_a=\ketbra{a^x}{a^x}\}_{x,a}, W=\{(x,v,a,b)\mid a=v\}_{x,v,a,b}\right),
\end{equation}
with $\mathcal X=\mathcal V=\mathcal A=\mathcal B=\{0,1\}^n$. 
Then, 
\begin{equation}
    \omega_{Q}(  \mathcal G^{\text{A } n \text{ guess}})=\frac{1}{2^n}\sum_{x,a}\tr{\bigg(V^{x}_{v}\otimes A^x_{a} \otimes (\sum_bB^x_{b}) \bigg)\rho_{RAB}}=\frac{1}{2^n}\sum_{x,a}\tr{(V^{x}_{v}\otimes A^x_{a} \otimes \mathbb I )\rho_{RAB}}.
\end{equation}
Notice that this game can be won with probability 1 if Alice and Bob send to the referee the maximally entangled state $\ket{\Phi^+}^{\otimes n}_{RA}$, Alice performs the same measurement as the referee and Bob answers at random:
\begin{equation}
    \omega_{Q}(  \mathcal G^{\text{A } n \text{ guess}})=1.
\end{equation}
Now, consider its constrained version given by imposing that Bob also has to guess the referee's measurement outcome up to an error $p_{err}$ per bit, but does not have to coordinate his answers with Alice. This variant will be motivated by device-independent QKD, see Section~\ref{sec:QKD}. In such a case, the set of constraints $C_{p_{err}}$ is described as: for every  strategy $S_{Q}=\{\rho_{RAB},A^x_a,B^y_b\}_{x,y,a,b}$,
\begin{equation}
    \sum_{x,a}q(x)\tr{(V^{x}_{v}\otimes A^x_{a} \otimes B^x_{b}) \rho_{RAB}} \leq p_{err}^{d_H(v,b)}\hspace{0.5cm}\forall v,b.
\end{equation}

In order to find an upper bound on the winning probability, we find the values of  $\omega_{admiss}^{(k=1)}(\mathcal{G}_C^{\text{A } n \text{ guess}})$ for $n=1,2$ computing the corresponding SDPs, see Fig.~\ref{Fig:QKD} for their values. 
Notice that we obtain tight bounds for $p_{err}=0$ and $p_{err}\geq\frac{1}{2}$, since for $p_{err}=0$, i.e.\ if Bob's outcome has to perfectly match the referee's outcome, the optimal values are upper bounded by $0.5$ and $0.25$, for $n=1$ and $n=2$, respectively, which can be obtained by the strategy where Alice and Bob send $\ket{\Phi^+}^{\otimes n}_{RB}$, Bob measures $V^x_b$ and Alice and makes the random guess $a=0\ldots0$, which will be correct with probability $\frac{1}{2^n}$.  On the other hand, if $p_{err}\geq \frac{1}{2}$, Bob can answer randomly and Alice can share $\ket{\Phi^+}^{\otimes n}_{RA}$ with the referee and perform the same measurement, in which case Alice's outcomes will be perfectly correlated with the referee's outcomes.

\begin{figure}[h]
\centering
\includegraphics[width=125mm]{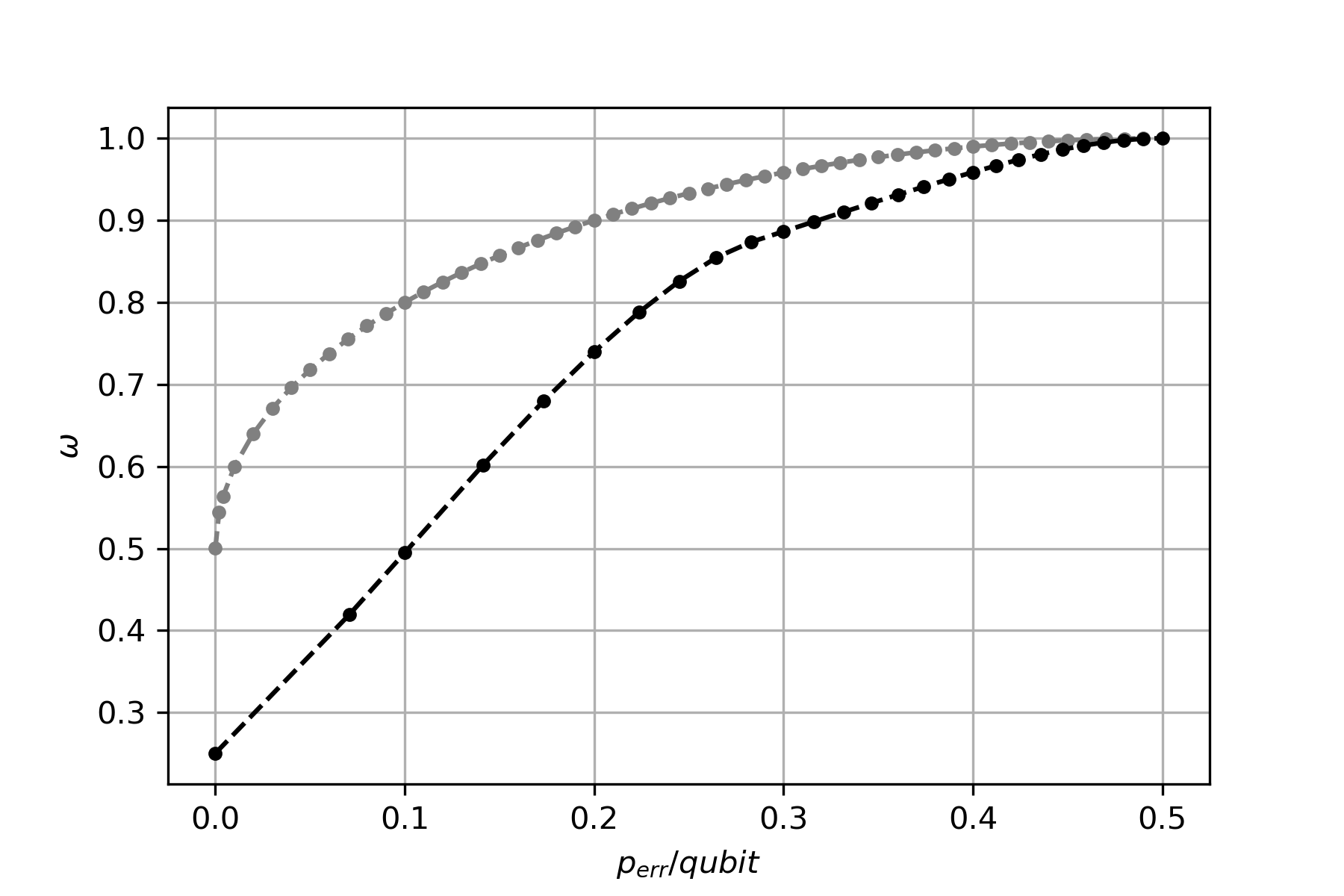}
\caption{Upper bounds for $\omega_Q(\mathcal{G}_{C_{p_{err}}}^{\text{A } n \text{ guess}})$ obtained from solving the SDP giving the value of $\omega_{admiss}^{(k=1)}(\mathcal{G}_C^{\text{A } n \text{ guess}})$ for $n=1$ (in grey) and $n=2$ (in black).}
\label{Fig:QKD}
\end{figure}

\subsection{Local guessing game}\label{section:local_guess_game}
Let $n\in\mathbb N$, and consider the extended non-local game where Alice and Bob receive input $x\in\{0,1\}$ and the referee chooses $z\in\{0,1\}^n$, all uniformly at random. The referee's measurements (on a system of $n$ qubits) are $V^z_v=\ketbra{v_0^{z_0}}{v_0^{z_0}}\otimes \ldots\otimes \ketbra{v_{n-1}^{z_{n-1}}}{v_{n-1}^{z_{n-1}}}$, for $v\in\{0,1\}^n$. Alice and Bob have to answer $a,b\in\{0,1\}^n$, and they win the game if and only if both answers are the same and, for every $i$ such that $z_i=x$, then $a_i=b_i=v_i$, i.e.\ whenever $R$ measures qubit $i$ in basis $x$, Alice and Bob have to guess correctly the measurement outcome. This game will be denoted by $\mathcal{G}^{n-\text{local guessing}}$, and, in the above notation,
\begin{equation}
    \mathcal{G}^{n-\text{local guessing}}:=\left(q(z,x,y)=\frac{\delta_{x,y}}{2}\frac{1}{2^n},V=\{V^z_v\}_{z,v}, W=\{(z,x,v,a,b)\mid a=b \text{ and } \forall i \text{ s.t.\ } z_i=x, a_i=v_i\}\right).
\end{equation}

\begin{prop} 
The optimal average winning probability of $\mathcal{G}^{n-\text{local guessing}}$ for $n=1$ is given by
\begin{equation}
    \omega_Q(\mathcal{G}^{1-\text{local guessing}})=\frac{1+\cos^2\left(\frac{\pi}{8}\right)}{2}.
\end{equation}
Moreover, this value is attained by the strategy $S_{Q}^{BB84}=\{\ketbra{\phi}{\phi},A_a^x=\delta_{a0},B_a^x=\delta_{a0}\}$, where $\ket{\phi}=\cos\frac{\pi}{8}\ket{0}+\sin\frac{\pi}{8}\ket{1}$.
\end{prop}
\begin{proof}
The strategy $S_{Q}^{BB84}=\{\ketbra{\phi}{\phi},A_a^x=\delta_{a0},B_a^x=\delta_{a0}\}$ is such that its winning average probability is $\frac{1+\cos^2\left(\frac{\pi}{8}\right)}{2}.$ On the other hand, consider an arbitrary strategy $\mathcal{S}=\{\rho_{RAB}, A^x_a, B^x_b\}$, then

\begin{equation}
\begin{split}
   \omega_{Q}&(\mathcal{G}^{1-\text{local guessing}})\big|_{\mathcal{S}}=
    \frac{1}{2^2}\sum_{(z,x,v,a,b)\in W}\tr{(V^z_v\otimes A^x_a\otimes B^x_b)\rho_{RAB}}\\&=
    \frac{1}{2^2}\sum_{x,v,a}\tr{(V^x_v\otimes A^{1-x}_a\otimes B^{1-x}_a)\rho_{RAB}}+\frac{1}{2^2}\sum_{x,a}\tr{(V^x_a\otimes A^x_a\otimes B^x_a)\rho_{RAB}}\\&=
    \frac{1}{2^2}\sum_{x,v}\tr{\big(V^x_v\otimes (\sum_a A^{1-x}_a\otimes B^{1-x}_a)\big)\rho_{RAB}}+\frac{1}{2^2}\sum_{x,a}\tr{(V^x_a\otimes A^x_a\otimes B^x_a)\rho_{RAB}}\\
    &\leq \frac{1}{2^2}\tr{\bigg(\sum_{x,v}V^x_v\bigg)\rho_{RAB}}+\frac{1}{2}\omega(\mathcal{G}^{BB84})=\frac{1}{2^2}\tr{2\rho_{RAB}}+\frac{1}{2}\omega(\mathcal{G}^{BB84})=\frac{1+\cos^2\left(\frac{\pi}{8}\right)}{2},
\end{split}\end{equation}
where we split the first sum in $x=z$ and $x\neq z$, then we used $\sum_a A^{1-x}_a\otimes B^{1-x}_a\preceq \mathbb I$ and that $\frac12\sum_{x,a}\tr{(V^x_a\otimes A^x_a\otimes B^x_a)\rho_{RAB}}\leq \omega_{Q}(\mathcal{G}^{BB84})$, and finally, that $\sum_{x,v}V^x_v=2\mathbb I$ and $\tr{\rho_{RAB}}=1.$
    
\end{proof}

The $n$-fold parallel repetition of the strategy $\mathcal S_{Q}^{BB84}$, denoted by $(\mathcal S_{Q}^{BB84})^{\times n}$ does not provide the optimal value for $\mathcal{G}^{n-\text{local guessing}}$ for all $n$. To see this, notice that on the one hand, this strategy has average winning probability 
\begin{equation}\label{eq:BC_breidbart_w}
    \omega(\mathcal{G}^{n-\text{local guessing}})\big|_{(S_{Q}^{BB84})^{\times n}}=\left(\frac{1}{4}\Big(3+\frac{1}{\sqrt{2}}\Big)\right)^n.
\end{equation}
On the other hand, consider the strategy $\mathcal{S}_{n-0}=\{\ket{0}^{\otimes n}, A^x_a=\delta_{a,0\ldots 0}, B^y_b= \delta_{b,0\ldots 0}\}$. This strategy is such that 
\begin{equation}
      \omega(\mathcal{G}^{n-\text{local guessing}})\big|_{\mathcal S_{n-0}}=\frac{1}{2}+\frac{1}{2}\left(\frac{3}{4}\right)^n.
\end{equation}
To see this, let $\mathcal{T}_{xz}:=\{i\mid z_i=x\}$, let $t=\abs{T_{xz}}$, denote by $\mathcal{T}_{xz}^c$ its complement, and, abusing of notation, we will write $\sum_{v\in\mathcal{T}_{xy}}$ for a bit string $v\in\{0,1\}$  to sum the indices $i$ of $v$ such that $i\in \mathcal{T}_{xz}$. Then, 
\begin{equation}
    \begin{split}
         \omega&(\mathcal{G}^{n-\text{local guessing}})\big|_{\mathcal S_{n-0}}=\frac{1}{2^{n+1}}\sum_{(z,x,v,a,b)\in W}\tr{(\ketbra{0}{0})^{\otimes n}\ketbra{v^z}{v^z}\delta_{a,0}}\\&=\frac{1}{2^{n+1}}\sum_{z,x}\tr{(\ketbra{0_{\mathcal T_{xz}}}{0_{\mathcal T_{xz}}}\otimes \ketbra{0_{\mathcal T_{xz}^c}}{0_{\mathcal T_{xz}^c}})(\ketbra{0_{\mathcal T_{xz}}^{x\ldots x}}{0_{\mathcal T_{xz}}^{x\ldots x}}\otimes  \sum_{v\in \mathcal T_{xz}^c}\ketbra{v_{\mathcal T_{xz}^c}^{1-x\ldots1-x}}{v_{\mathcal T_{xz}^c}^{1-x\ldots1-x}})}\\
         &=\frac{1}{2^{n+1}}\sum_{z,x}\tr{(\ket{0_{\mathcal T_{xz}}}\braket{0_{\mathcal T_{xz}}}{0_{\mathcal T_{xz}}^{x\ldots x}}\bra{0_{\mathcal T_{xz}}^{x\ldots x}}\otimes
         \ketbra{0_{\mathcal T_{xz}^c}}{0_{\mathcal T_{xz}^c}}\mathbb I}\\
         &=\frac{1}{2^{n+1}}\left(\sum_{z,x-0}\tr{(\ket{0_{\mathcal T_{xz}}}\braket{0_{\mathcal T_{xz}}}{0_{\mathcal T_{xz}}}\bra{0_{\mathcal T_{xz}}}}+\sum_{z,x=1}\tr{(\ket{0_{\mathcal T_{xz}}}\braket{0_{\mathcal T_{xz}}}{0_{\mathcal T_{xz}}^{1\ldots 1}}\bra{0_{\mathcal T_{xz}}^{1\ldots 1}}}\right)\\
         &=\frac{1}{2^{n+1}}\left(\sum_z1+\sum_{t=0}^n\binom{n}{t}2^{-t}\right)=\frac{1}{2}+\frac{1}{2}\left(\frac{3}{4}\right)^n,
    \end{split}
\end{equation}
where we used that $\tr{(\ket{0_{\mathcal T_{xz}}}\braket{0_{\mathcal T_{xz}}}{0_{\mathcal T_{xz}}^{1\ldots 1}}\bra{0_{\mathcal T_{xz}}^{1\ldots 1}}}=\abs{\braket{0_{\mathcal T_{xz}}}{0_{\mathcal T_{xz}}^{1\ldots 1}}}=2^{-t}$. Similarly, one can see \eqref{eq:BC_breidbart_w}.

We see that for $n\geq 8$, the strategy $\mathcal{S}_{n-0}$ outperforms $(\mathcal S_{Q}^{BB84})^{\times n}$. Not only this, but the strategy $(\mathcal S_{Q}^{BB84})^{\times n}$ decays exponentially in $n$, whereas $\mathcal{S}_{n-0}$ has an average winning value of at least $\frac{1}{2}$. Therefore, we can conclude that 
\begin{equation}
     \omega_Q(\mathcal{G}^{n-\text{local guessing}})\geq\frac{1}{2}+\frac{1}{2}\left(\frac{3}{4}\right)^n.
\end{equation}

In \cite{Pital_a_Garc_a_2018} an upper bound  was shown  that can be directly transformed in an upper bound for $\mathcal{G}^{n-\text{local guessing}}$, which we state in the below proposition:

\begin{prop} \label{prop:bound_local_guess} For every $n\in\mathbb N$, the following bound holds
\begin{equation} \label{eq:bound_local_guess}
    \omega_Q(\mathcal{G}^{n-\text{local guessing}})\leq \frac{1}{2}+\frac{1}{2}\left(\frac{1+1/\sqrt{2}}{2}\right)^n. 
\end{equation}
\end{prop}

The proof in \cite{Pital_a_Garc_a_2018} is based on similar techniques as in the proof of Theorem 3. in \cite{TomamichelMonogamyGame2013}. For completeness, we show the proof in terms of extended non-local games in the Appendix.

In order to verify tightness of \eqref{eq:bound_local_guess}, we compute the SDPs providing the values of $\omega_{admiss}^{(k=1)}(\mathcal{G}^{n-\text{local guessing}})$ for $n=1,2,3,4,5$, see Fig.~\ref{Fig:local_guess_game}. We see that using SDPs, we obtain tighter upper bounds for the above mentioned values of $n$. 

\begin{figure}[h]
\centering
\includegraphics[width=105mm]{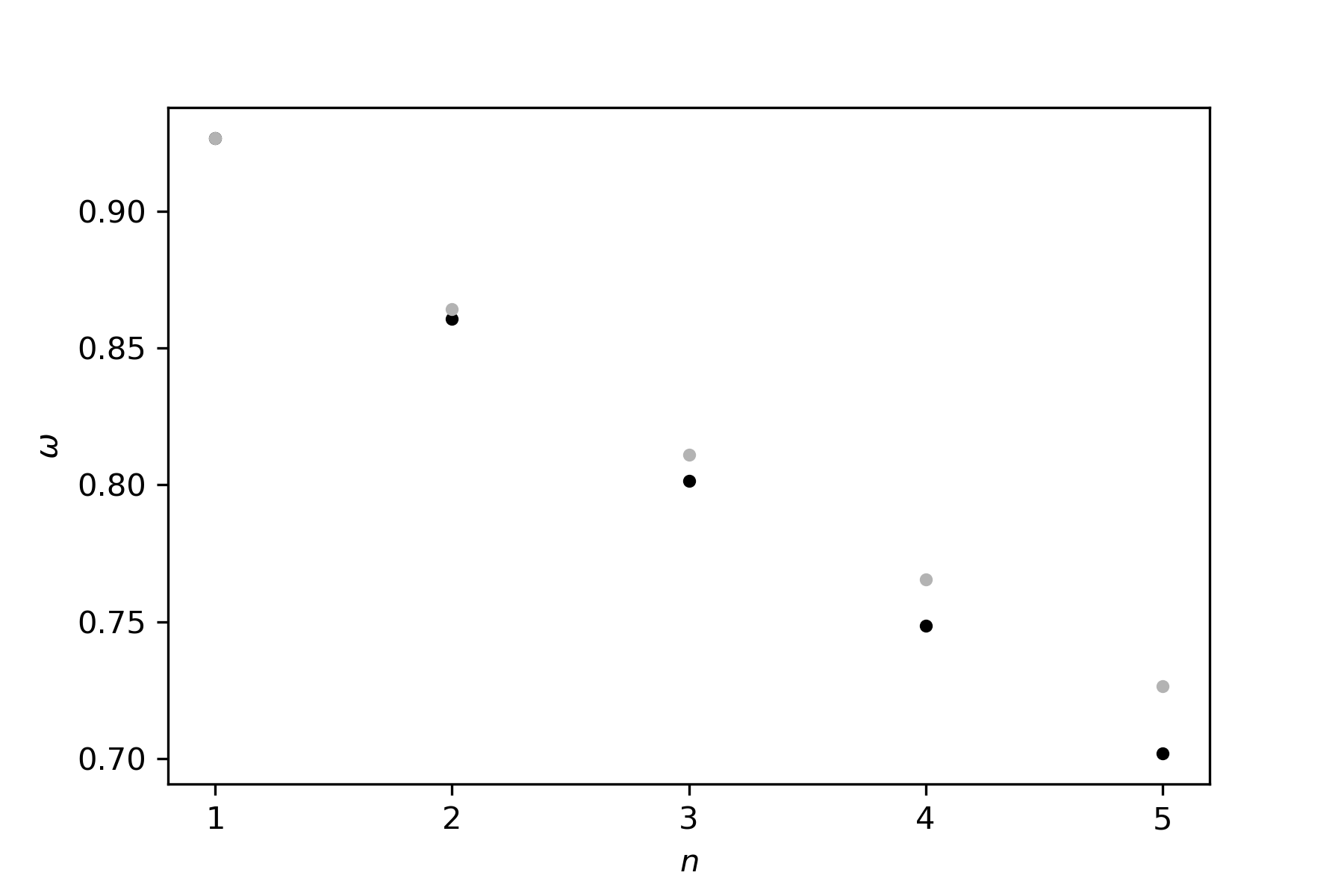}
\caption{Upper bounds in Proposition \ref{prop:bound_local_guess}, given by \eqref{eq:bound_local_guess}, (gray dots) and upper bounds obtained by $\omega_{admiss}^{(k=1)}(\mathcal{G}^{n-\text{local guessing}})$ (black dots).}
\label{Fig:local_guess_game}
\end{figure}

\subsection{Lossy monogamy-of-entanglement games}\label{section:lossy_moe_games}

Consider the lossy-and-constrained version of $\mathcal{G}^{BB84}$, denoted by $\mathcal{G}_{C,\eta}^{BB84}$, where the set of constraints $C$ is given as in the previous case by imposing that Alice and Bob never answer differently, i.e.\ imposing that their probability of answering different is 0. In this case, the set of constraints $C$ is given by: for every  strategy $S_{Q}=\{\rho_{RAB},A^x_a,B^y_b\}_{x,y,a,b}$,
\begin{equation}
    \sum_{x,a,a'\neq b'}q(x)\tr{(V^{x}_{a}\otimes A^x_{a'} \otimes B^x_{b'}) \rho_{RAB}}=0.
\end{equation}
In \cite{Escol_Farr_s_2023} tight results for the optimal winning probability of this lossy-and-constrained game were found by an \emph{ad hoc} method combining the `$1+AB$' level of the NPA hierarchy \cite{NPA2008} with extra linear constraints derived from bounding operator norms. 
We show that solving the SDP \eqref{eq:SDP} for $k=`1+AB'$ for $\mathcal{G}_{C,\eta}^{BB84}$,  see code \cite{code}, it confirms the results in~\cite{Escol_Farr_s_2023}, showing that the optimal value is achieved by the strategy given by the convex combination of $S_{Q}^{BB84}$ and $S_{Q}^{guess}=\{\ketbra{0}{0},A_a^0=\delta_{a0},A_a^1=\delta_{a\perp},B_a^0=\delta_{a0},B_a^1=\delta_{a\perp}\}$, which is the strategy where Alice and Bob make the guess $x=0$ for the referee's measurement, and they answer $0$ if the guess was right, and they answer $\perp$ if their guess was wrong. See Fig.~\ref{Fig_bb84_eta} for a plot of the results. We see that level 1, $\omega_{admiss}^{(k=1)}$, provides a good approximation, whereas level 1 in \cite{Escol_Farr_s_2023} is far from the optimal value. The main advantage is that it is enough to have the description of the game to find the value and there is no need to derive extra linear constraints. Moreover, we know that $\omega_{admiss}^{(k)}$ will converge to the true value, whereas it was not necessarily the case in the \emph{ad hoc} method in \cite{Escol_Farr_s_2023}. This becomes more clear in the following case. 

\begin{figure}[h]
\centering
\includegraphics[width=125mm]{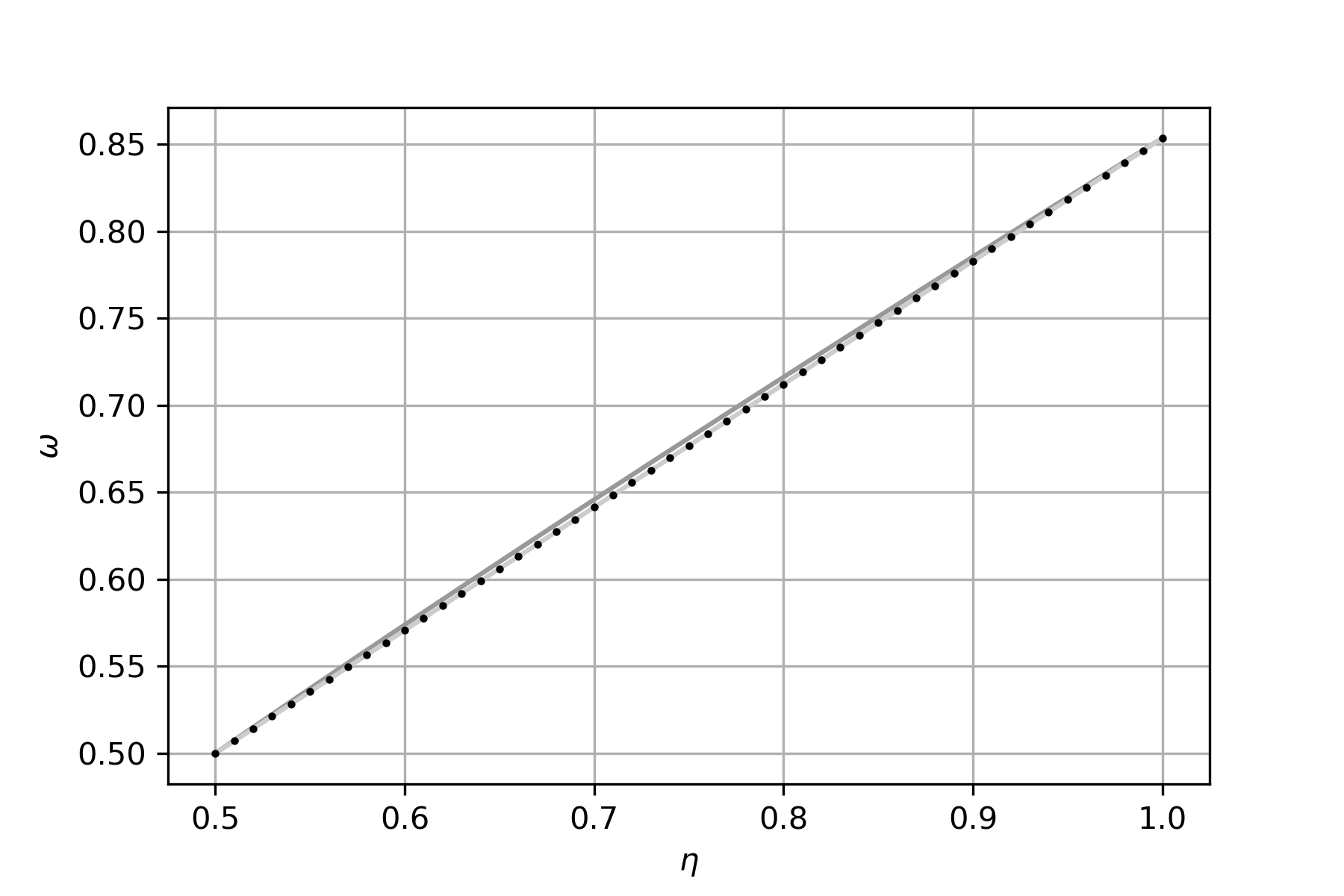}
\caption{Solutions of $\omega_{admiss}^{(k=1)}(\mathcal{G}_{C,\eta}^{BB84})$ and $\omega_{admiss}^{(k=`1+AB')}(\mathcal{G}_{C,\eta}^{BB84})$, represented by dark grey continuous line and black dots, respectively, together with the winning probability of the strategy given by the convex combination of the strategies  $S_{Q}^{BB84}$ and $S_{Q}^{guess}$, represented by a light grey continuous line. }
\label{Fig_bb84_eta}
\end{figure}

Consider the extended non-local game $\mathcal{G}^{3-bases}$ described by
\begin{equation}
    \mathcal{G}^{3-bases}:=\left(q(x)=\frac{1}{\abs{\mathcal{X}}},V=\{V^{x}_{a}=\ketbra{a_x}{a_x}\}_{x,a}, W=\{(x,v,a,b)\mid v=a=b\}_{x,v,a,b}\right),
\end{equation}
with $\mathcal X=\{0,1,2\}$ and $\mathcal A=\{0,1\}$, where $\ket{a_0}=\ket{a}$, $\ket{a_1}=H\ket{a}$, and $\ket{a_2}=\ket{i}$ if $a=0$ and $\ket{a_2}=\ket{-i}$ if $a=1$. This game is similar to $\mathcal{G}^{BB84}$ but the referee, instead of measuring his local register in either the computational or the Hadamard basis, can also measure in the basis $\{\ket{\pm i}\}$. The lossy-and-constrained version of this game, $\mathcal{G}_{C,\eta}^{3-bases}$, where $C$ is given by forbidding different answers from Alice and Bob,  was analyzed in \cite{Escol_Farr_s_2023} with the \emph{ad hoc} method providing upper bounds on $\omega_Q(\mathcal{G}_{\eta}^{3-bases})$ for $\forall\eta\in[0,1]$.  Nevertheless, it was open if those were tight. Here we show that solving the SDPs for $\omega_{admiss}^{k}(\mathcal{G}_{C,\eta}^{3-bases})$ gives better upper bounds, see Fig. \ref{Fig:3_bases}. In addition, we will see that they are tight. Moreover, we computed the values for  $\omega_{admiss}^{k=1}(\mathcal{G}_{\eta}^{3-bases})$, i.e.\ without the constraint of imposing same answers for Alice and Bob, and we obtained the same value. In search of optimal strategies for every $\eta$, this fact lead us to think of strategies that always coordinate the answers of Alice and Bob, consisting on preparing a concrete qubit and pre-agreeding a fix answer regardless the question they receive. 

Consider the strategies 
\begin{equation}
    S_Q^{3-guess}=\{\ketbra{0}{0},A_a^0=\delta_{a0},A_a^1=A_a^2=\delta_{a\perp},B_a^0=\delta_{a0},B_a^1=B_a^2=\delta_{a\perp}\},
\end{equation}
which consists on guessing $x=0$ and answering $\perp$ (photon loss) if the guess was wrong, 
\begin{equation}
    S_{Q}^{3-BB84}=\{\ketbra{\phi}{\phi},A_a^0=A_a^1=\delta_{a0},A_a^2=\delta_{a\perp},B_a^0=B^1_a=\delta_{a0},B_a^2=\delta_{a\perp}\},
\end{equation}
consisting on using $S_{Q}^{BB84}$ if $x\in\{0,1\}$ and claiming photon loss if $x=2$,
and 
\begin{equation}
    S_{Q}^{3-bases}=\{\ketbra{\phi_3}{\phi_3},A_a^x=\delta_{a0},B_a^x=\delta_{a0}\},
\end{equation}
where   $  \ket{\phi_3}=\cos\left(\frac{\tan^{-1}(\sqrt{2})}{2}\right)\ket{0}+e^{i\frac{\pi}{4}}\sin\left(\frac{\tan^{-1}(\sqrt{2})}{2}\right)\ket{1}$. Notice that $\ket{\phi_3}$ is the state that has simultaneous maximum overlap with $\ket0,\ket+$ and $\ket{i}$, i.e. 
\begin{equation}
    \max_{\ket\varphi}\left\{\abs{\braket{0}{\varphi}}^2+\abs{\braket{+}{\varphi}}^2+\abs{\braket{i}{\varphi}}^2\right\}=\abs{\braket{0}{\phi_3}}^2+\abs{\braket{+}{\phi_3}}^2+\abs{\braket{i}{\phi_3}}^2.
\end{equation}

\begin{result}\label{result:3_bases}
    The strategy $\mathcal{S}_{Q}^{3\eta}$ consisting of Alice and Bob playing:
    \begin{itemize}
        \item the convex combination of $ S_Q^{3-guess}$ and $S_{Q}^{3-BB84}$, for $\eta\in[\frac{1}{3},\frac{2}{3})$, 
        \item the convex combination of  $S_{Q}^{3-BB84}$ and $ S_Q^{3-bases}$, for $\eta\in[\frac{2}{3},1]$ 
    \end{itemize}
     is optimal for $\mathcal{G}_{\eta}^{3-bases}$ with
     \begin{equation}
         \omega_Q(\mathcal{G}_{\eta}^{3-bases})=\begin{cases}
             &\frac{1}{3}+\frac{1}{\sqrt{2}}(\eta-\frac{1}{3}) \text{ if } \eta\in[\frac{1}{3},\frac{2}{3}),\\
             &\left(\frac{1}{\sqrt{2}}-\frac{1}{\sqrt{3}}\right)+\frac{1-\sqrt{2}+\sqrt{3}}{2}\eta \text{ if } \eta\in[\frac{2}{3},1].
         \end{cases}
     \end{equation}
\end{result}
The strategy  $\mathcal{S}_{Q}^{3\eta}$ matches the upper bound obtained by the SDP giving the value of $\omega_{admiss}^{(k)}(\mathcal{G}_{\eta}^{3-bases})$, where $k=`1+AB'$ for $\eta\in[\frac{1}{3},\frac{2}{3})$ and $k=1$  for $\eta\in[\frac{2}{3},1]$. We want to highlight that $k=1$ for $\eta\in[\frac{1}{3},\frac{2}{3})$ has solutions that are almost equal for $k=`1+AB'$ and just diverge in the third decimal. Notice that $\mathcal{S}_{Q}^{3\eta}$ slightly outperforms the strategy given by the convex combination of $ S_Q^{3-guess}$ and $S_{Q}^{3-bases}$, see blue line in Fig.~\ref{Fig:3_bases}. 

\begin{figure}[h]
\centering
\includegraphics[width=125mm]{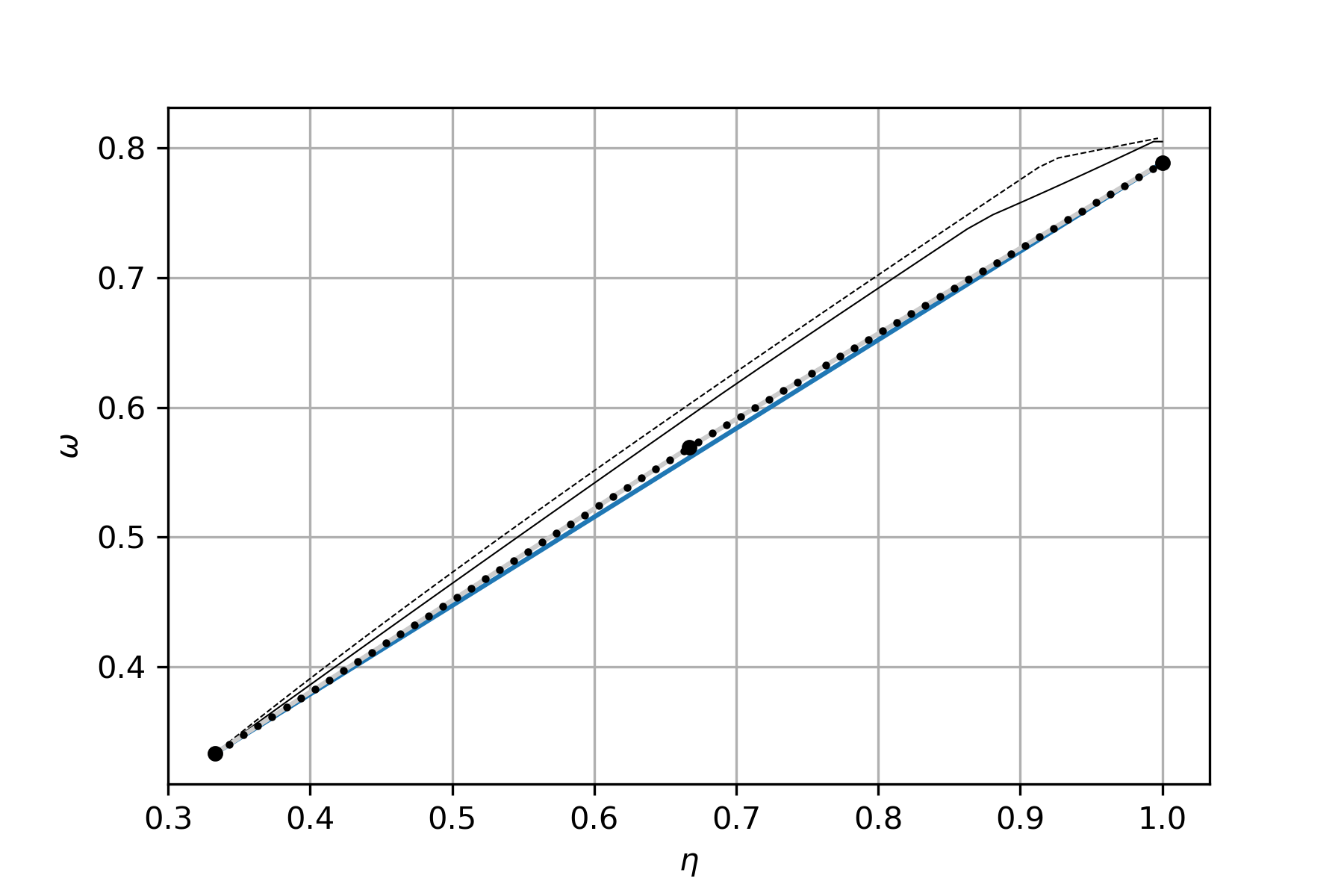}
\caption{The dotted and continuous lines correspond to the upper bounds obtained in \cite{Escol_Farr_s_2023} with the \emph{ad hoc} method using the $1$st and $2$nd levels of the NPA hierarchy, respectively. The black dots correspond to the values of $\omega_{admiss}^{(k)}(\mathcal{G}_{\eta}^{3-bases})$, where $k=`1+AB'$ for $\eta\in[\frac{1}{3},\frac{2}{3})$ and $k=1$  for $\eta\in[\frac{2}{3},1]$ , the gray line correspond to the winning values of the strategy   $\mathcal{S}_{Q}^{3\eta}$, the blue line corresponds to the winning value of the strategy given by the convex combination of $ S_Q^{3-guess}$ and $S_{Q}^{3-bases}$, and the big black dots correspond to the winning probabilities of the strategies $ S_Q^{3-guess}$, $S_{Q}^{3-BB84}$ and  $ S_Q^{3-bases}$, respectively.}
\label{Fig:3_bases}
\end{figure}

Unlike for $\mathcal{G}^{BB84}_{\eta}$, the optimal strategy for $\mathcal{G}_{\eta}^{3-bases}$ for every $\eta$ is not just given by simply playing the optimal strategy (for $\eta=1$) with a certain frequency and combining it with the guessing strategy, claiming photon loss enough times to be consistent with $\eta$. This is summarized in the below observation:

\begin{observation}
    The optimal winning probability of a lossy extended non-local game is not always given by the convex combination of the optimal strategy for $\eta=1$ and the guessing attack.
\end{observation}

\section{Applications to quantum cryptography}\label{sec:applications}
\subsection{Application to relativistic Bit Commitment (with loss)}\label{sec:rel_BC}

Bit commitment (BC) is a cryptographic two-party primitive where a party Charlie commits to a bit $x\in\{0,1\}$ relative to another party $R$ in a way that is both binding, in the sense that Charlie can not change the value of $x$ afterward, and hiding, in the sense that $R$ cannot know the value of $x$ before Charlie unveils it.  The standard description of bit commitment protocols consist of the following phases: commit (Charlie commits to $x$), wait (no information is shared between parties), and open (Charlie unveils $x$ to R). In \cite{PhysRevLett.78.3410,PhysRevLett.78.3414} the impossibility of unconditional security of bit commitment was shown, even for protocols constructed using quantum information.  

In \cite{Kent_2012_BC}, Kent introduced a BC protocol, which we denote by BC$^{n-BB84}$, in which Charlie is forced by the spacetime constraints in a Minkowski spacetime to commit $x$ at a given spacetime point. Unconditional security in the asymptotic limit (in $n$, see below) was proven in \cite{Kaniewski_2013}, with tighter results for every $n$ in \cite{Pital_a_Garc_a_2018}. In this section, we show that security of this protocol can be reduced to an extended non-local game and, by computing upper bounds as in Section~\ref{section:convergence_proof}, we find tighter results (for small values of $n$). We show security for small instances of $n$ when loss is taken into consideration.

The BC$^{n-BB84}$ protocol is described as follows: let $(s,t)$ be a set of coordinates for $(1+1)$-dimensional Minkowski spacetime and consider the points $P=(0,0)$, $Q_0=(-1,1)$ and $Q_1=(1,1)$. Notice no information sent from $P_t=(0,t)$ for $t>0$ can reach both $Q_0$ and $Q_1$. The referee $R$ prepares the state $\ket\psi=\ket{\psi_1}\otimes \ldots \otimes \ket{\psi_n}$, where $\ket{\psi_i}\in\{\ket0,\ket1,\ket+,\ket-\}$ are BB84 states, and sends $\ket\psi$ to Charlie, who, at $P$, commits the bit $x$ by measuring each qubit he received in basis $x\in\{0,1\}$, obtaining $b\in\{0,1\}^n$. Charlie sends $x,b$ to two agents, Alice and Bob, who unveil $x$ (and $b$) at $Q_0$ and $Q_1$, respectively. The referee checks if both $x,b$ from $Q_0$ and $Q_1$ are the same and are consistent with the states he prepared. 

Security for dishonest referee and honest Charlie follows from the fact that the referee cannot learn anything about $x$ by no-signalling.

In order to study security for honest referee and dishonest Charlie (and his agents), we introduce a purified version of  BC$^{n-BB84}$, in a similar way as in \cite{Kaniewski_2013}, described as follows:
\begin{enumerate}
    \item $R$ prepares $n$ EPR pairs $\otimes_{i=1}^n\ket{\Phi^+}_{R_iC_i}$ and sends one half of each pair to Charlie ($C$), i.e.\ registers $C_1\ldots C_n$ and keeps the other halves (registers $R_1\ldots R_n$).
    \item Charlie ($C$) commits to a bit $x$ by measuring the qubits he received in basis $x\in\{0,1\}$, obtaining $b\in\{0,1\}^n$ (committing phase), and broadcasts $x$ and $b$ to two non-communicating agents Alice ($A$) and Bob ($B$).
    \item Alice and Bob send $x$ and $b$ to R (opening phase).
    \item R picks a uniformly random bit string $z\in\{0,1\}^n$ and measures the qubit in register $R_i$ in the basis $z_i$ for all $i\in[n]$, obtaining $r\in\{0,1\}^n$. 
    \item R checks if  
    \begin{itemize}
        \item he received the same $x$ and $b$ from Alice and Bob,
        \item for all $i$ such that $z_i=x$, it should hold that $r_i=b_i$. We will denote $(x,z,r,b)$ such that the above hold as $W$, the \emph{winning set}. 
    \end{itemize} 
    If the checks pass, the opening is accepted. 
\end{enumerate}
The only difference with the purification of BC$^{n-BB84}$ in \cite{Kaniewski_2013} is that, in our case, the referee measures each qubit in a uniformly random basis $z_i\in\{0,1\}$ instead of measuring half of the $n$ qubits in the computational basis and the other half in the Hadamard basis. Our purification corresponds to Kent's protocol \cite{Kent_2012_BC}, since all the BB84 states are in a uniformly random basis, instead of half in computational and half in Hadamard as in the variation in \cite{Kaniewski_2013}.

The average probability of accepting the commitment is given by 
\begin{equation}
    \omega_{accept}=\frac{1}{2}(p_0+p_1),
\end{equation}
where $p_0$ and $p_1$ denote the probability of accepting the commitment $x=0$ and $x=1$, respectively. For an honest (and perfect) implementation of the protocol,
\begin{equation}
    p_{x}=\frac{1}{2^n}\sum_{(x,z,r,b)\in W}\tr{(V^{z_0}_{r_0}\otimes\cdots\otimes V^{z_{n-1}}_{r_{n-1}})\otimes (V^{x}_{b_0}\otimes\cdots \otimes V^{x}_{b_{n-1}})(\ket{\Phi^+} \bra{\Phi^+})^{\otimes n}}, \text{ for } x\in\{0,1\},
\end{equation}
where $V^{y}_{a}=H^{y}\ketbra{a}{a}H^y$ for $a,y\in\{0,1\}$. Then, the average probability of accepting is 
\begin{equation}
    \omega_{accept}^{honest}=\frac{1}{2}\sum_x p_x=1.
\end{equation}

See Fig.~\ref{fig:BC_scheme}~(a) for a schematic representation of BC$^{n-BB84}$. In Fig.~\ref{fig:BC_scheme}, we place $R$ in the same location as $B'$ in the past, but $R$ can be anywhere in the past of $C$. 

\begin{figure}[H]
\centering
\begin{subfigure}{0.48\textwidth}
    \centering
    \includegraphics[width=\textwidth]{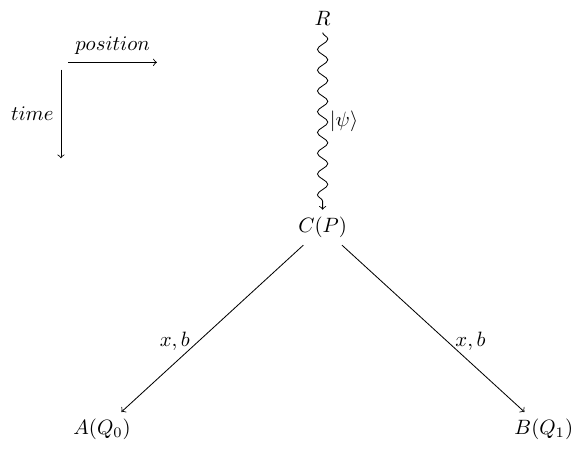}
    \caption{Honest implementation}
\end{subfigure}
\hfill
\begin{subfigure}{0.48\textwidth}
    \centering
    \includegraphics[width=\textwidth]{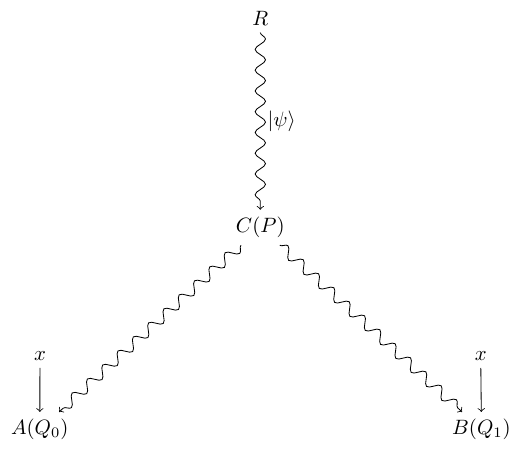}
    \caption{Dishonest implementation}
\end{subfigure}
\caption{Scheme of an honest (a) and dishonest (committer) implementation (b) of the BC$^{n\text{-}BB84}$ protocol.}
\label{fig:BC_scheme}
\end{figure}

When transmitting the quantum information from $R$ to $C$, for example, by sending it through a quantum channel connecting $R$ to $C$, or by $R$ `delivering' it to $C$ such that $C$ has to keep it in a quantum memory till the committing phase, a certain fraction of qubits might be lost. Let $\eta$ denote the transmission rate per qubit, then with probability $1-\eta$, in the commitment phase Charlie will obtain $b_i=\perp$ for all $i\in[n]$. We introduce BC$_{\eta}^{n-BB84}$  as the lossy version of BC$^{n-BB84}$ with transmission rate per qubit $\eta$, where the difference is that $b\in\{0,1,\perp\}^n$ and the check that R performs consists of
\begin{itemize}
    \item verifying that he received the same $x$ and $b$ from Alice and Bob,
    \item verifying consistency with the measurement outcomes, i.e. 
    \begin{equation}\label{eq:BC_consistency_answers}
        \text{for all }i \text{ such that }z_i=x,\text{ and }b_i\neq\perp, \text{ }r_i=b_i,
    \end{equation}
    \item verifying consistency with the transmission rate, i.e. 
    \begin{equation}\label{eq:BC_consistency_eta}
        \abs{\{i\mid b_i\neq\perp\}}\in [n\eta-\sigma_{\eta},n\eta+\sigma_{\eta}],
    \end{equation}
     for a certain $\sigma_{\eta}>0$.
\end{itemize}
If the checks pass, the opening is accepted. 
The set of tuples $(x,z,r,b)$ that pass the test (winning set) is given by
\begin{equation}
    W_\eta=\{(x,z,r,b)\in\{0,1\}\times\{0,1\}^n\times\{0,1\}^n\times\{0,1,\perp\}^n\mid \text{ eqs. } \eqref{eq:BC_consistency_answers}\text{ and }\eqref{eq:BC_consistency_eta} \text{ hold.}\},
\end{equation}
i.e., if the answers she received are such that $(x,z,r,b)\in W_\eta$, the bit $x$ is accepted. 

For an honest execution of the protocol, we consider the lossy channel that acts on each qubit by leaving it invariant with probability $\eta$, and losing it with probability $1-\eta$. The honest average winning probability can be expressed as 
\begin{equation}
    \omega_{accept}^{honest}=\frac{1}{2}\frac{1}{2^n}\sum_{(x,z,r,b)\in W_\eta}\tr{(V^z_r\otimes M^x_b)(\ket{\Phi^+} \bra{\Phi^+})^{\otimes n}},
\end{equation}
where $M^x_{b}=M^x_{b_0}\otimes\cdots\otimes M^x_{b_{n-1}}$, for
\begin{equation}
    M^x_{b_i}=\begin{cases}
        \eta H^{x}\ketbra{b_i}{b_i}H^{x} &\text{ if } b\in\{0,1\},\\
        (1-\eta)\mathbb I_{2}&\text{ if } b=\perp.
    \end{cases}
    \hspace{ 0.5cm }\forall i\in[n],
\end{equation}
i.e.\ in each qubit Charlie either performs the same measurement as the referee with probability $\eta$, or the photon is lost with probability $1-\eta$. For $n$ large enough, by the central limit theorem, the commitment will be accepted with probability arbitrarily close to 1 (with the value determined by $\sigma_\eta$). 

Moreover, in an experimental implementation, an honest Charlie will have a certain error. Let $\varepsilon$ be the error per qubit that he is assumed to have. Then, taking this into account, we introduce the loss-and-error version of BC$^{n-BB84}$ as BC$_{\eta}^{n-BB84}$ where the check that R performs consists of the same check as in the lossy version and, in addition, she verifies the consistency with the error, i.e.,
    \begin{equation}\label{eq:BC_consistency_error}
        \varepsilon-\sigma_{\varepsilon}\leq \frac{\abs{\{i\mid z_i=x \text{ and } b_i\neq \perp \text{ and } b_i\neq r_i\}}}{\abs{\{i\mid z_i=x \text{ and } b_i\neq \perp \}}}\leq \varepsilon+\sigma_{\varepsilon},
    \end{equation}
     for a certain $\sigma_{\varepsilon}>0$.

\subsubsection{Security for dishonest Charlie}

We analyze security in the \emph{global command model} \cite{Kaniewski_2013}, where an external agent dictates the bit that should be unveiled to Alice and Bob at latest at $t=1$. The most general thing a dishonest Charlie can do is, during the commit phase, to apply an arbitrary completely positive trace-preserving (CPTP) map to his halves of EPR pairs and possibly some ancillary qubits and send a register to Alice and another register to Bob. Then, upon receiving the bit $x$ dictated by an external agent, Alice and Bob perform POVMs $\{A^x_a\}$ and $\{B^x_b\}$ to answer $x,a$ and $x,b$, respectively. 

Security of BC$_{\eta}^{n-BB84}$ can be reduced to the winning probability of the extended non-local game `Local guessing game' $\mathcal G^{n-\text{local guessing}}$, described in Section~\ref{section:local_guess_game}, where the referee, Alice and Bob share a tripartite state $\rho_{RAB}$, and Alice and Bob answer $x,a$ and $x,b$ such that $a=b$. Upon the choice of measurement $z\in\{0,1\}^n$ with outcome $r\in\{0,1\}^n$, they win the game if the tuple $(x,z,r,b)\in W_{\eta}$. 
The optimal winning probability is given by 
\begin{equation}\label{eq:attack_BC}
    \omega_{Q}=\sup\frac{1}{2}\frac{1}{2^n}\sum_{(x,z,r,b)\in W_\eta}\tr{(V^z_r\otimes A^x_b\otimes B^x_b)\rho_{RAB}}=\frac{1}{2}(p^{attack}_0+p^{attack}_1)=: \omega_{accept}^{attack},
\end{equation}
where $V^z_r=\ketbra{r_0^{z_0}}{r_0^{z_0}}\otimes\cdots\otimes \ketbra{r_{n-1}^{z_{n-1}}}{r_{n-1}^{z_{n-1}}}$, and $p_{x}^{attack}=\frac{1}{2^n}\sum_{(x,z,r,b)\in W}\tr{(V^z_r\otimes A^x_b\otimes B^x_b)\rho_{RAB}}$ denotes the probability that $R$ accepts the commit $x$ (sent by the attackers).

Notice that a dishonest implementation of the protocol will always have average winning probability of at least $\frac{1}{2}$, since a dishonest Charlie can pick $x$, measure all the qubits he received at $P$ in basis $x$ (as in an honest implementation) and broadcast the outcome to Alice and Bob. If the external agent picks the same value $x$, the commitment will be accepted, thus, 
\begin{equation}
    \omega_{accept}^{attack}=\frac{1}{2}+\xi,
\end{equation}
for a certain bias $\xi\geq0$. Phrasing the attack in terms of an extended non-local game, the strategy $\mathcal{S}_{n-0}=\{\ket{0}^{\otimes n}, A^x_a=\delta_{a,0\ldots 0}, B^y_b= \delta_{b,0\ldots 0}\}$ introduced in Section \ref{section:local_guess_game} is such that 
\begin{equation}
     \omega_{accept}^{attack}\big |_{\mathcal{S}_{n-0}}=\frac{1}{2}(p^{attack}_{x_0}+p^{attack}_{1-x_0})=\frac{1}{2}+\frac{1}{2}\left(\frac{3}{4}\right)^n.
\end{equation}
Therefore, we have that 
\begin{equation}
    \omega_{accept}^{attack}\geq \frac{1}{2}+\frac{1}{2}\left(\frac{3}{4}\right)^n.
\end{equation}
In \cite{Kaniewski_2013}, security was shown with
\begin{equation}
   \omega_{accept}^{attack}\leq\frac{1}{2} +\inf_{\delta\in(0,\frac{1}{2})}2^{-n(1-h(\delta))}+ e^{-\frac{1}{2}n\delta^2},
\end{equation}
which decays exponentially for $n\rightarrow\infty$. This result was tightened in \cite{Pital_a_Garc_a_2018}, showing that, from  Proposition~\ref{prop:bound_local_guess},
\begin{equation}
      \omega_{accept}^{attack}\leq\frac{1}{2} +\frac{1}{2}\left(\frac{1+1/\sqrt{2}}{2}\right)^n. 
\end{equation}

We analyze security for small values of $n$ by computing upper bounds of \eqref{eq:attack_BC} solving the SDP that gives us the upper bound $\omega_{admiss}^{(1)}$ for fixed values of $n$, see code \cite{code}. 

For the lossless and error-free case, we have the bounds plotted in Fig.~\ref{Fig:local_guess_game}, where we find tighter bounds than in \cite{Pital_a_Garc_a_2018} for small values of $n$.

Now, for small $n$, we consider the scenario where a certain number of photon loss answers are accepted. Since the acceptance criterion will be based on a few qubits, we only consider the error-free case.  We analyze the following cases:

\begin{enumerate}[label=\textit{(\roman*)}]
    \item $n=1$. The referee has to either accept or reject the commitment based on a single bit. In such a case, we consider the error-free case where only non no-photon answers are accepted, i.e.\ $W=\{(x,z,r,b)\mid \text{ if } z=x, \text{ then }r=b\neq\perp\}$. Then, 
    \begin{equation}
        \omega_{Charlie}=\eta \hspace{1cm}\text{ and } \hspace{1cm}\omega_{admiss}^{(1)}=0.926776695\simeq\frac{1+\cos^2\frac{\pi}{8}}{2}.
    \end{equation}
    
    From Section \ref{section:local_guess_game}, we know that the result $\omega_{admiss}^{(1)}$ is tight. In this case, an honest Charlie would pass the test with higher probability than potential attackers as long as $\eta>0.926776695$. There is no point in considering accepting photon loss as a valid answer for $n=1$, since the attackers can then always claim photon loss and be accepted with probability~1. 
    \item $n=2$. The referee has to either accept or reject the commitment based on two bits. In such a case, we consider two possible winning sets $W$ given by
    \begin{itemize}
        \item the set of two answers such that do not contain $\perp$ (no photon loss answers accepted), then 
        \begin{equation}
        \omega_{Charlie}=\eta^2 \hspace{1cm}\text{ and } \hspace{1cm}\omega_{admiss}^{(1)}=0.8607577 .
    \end{equation}
    In this case, an honest Charlie would pass the test with higher probability than potential attackers as long as $\eta>0.92777$. 
        \item the set of two answers such that at least one is not $\perp$. In such a case, $\omega_{admiss}^{(1)}=1.000$ provides a trivial upper bound.

    \end{itemize}
     \item $n=3$. The referee has to either accept or reject the commitment based on three bits. In such a case, we consider two possible winning sets $W$ given by
    \begin{itemize}
        \item the set of two answers such that do not contain $\perp$ (no photon loss answers accepted), then 
        \begin{equation}
        \omega_{Charlie}=\eta^3\hspace{1cm}\text{ and } \hspace{1cm}\omega_{admiss}^{(1)}=0.8014794 %\lle{ 0.8014794471123362}.
    \end{equation}
     In this case, an honest Charlie would pass the test with higher probability than potential attackers as long as $\eta>0.928892$. 
        \item the set of two answers such that at least one is not $\perp$, then
        \begin{equation}
        \omega_{Charlie}= (3-2\eta)\eta^2\hspace{1cm}\text{ and } \hspace{1cm}\omega_{admiss}^{(1)}=0.927699%\lle{0.9276990736900129}.
    \end{equation}
    In this case, an honest Charlie would pass the test with higher probability than potential attackers as long as $\eta>0.835472$.

     \end{itemize}
\end{enumerate}

\subsection{Application to device-independent quantum key distribution}\label{sec:QKD}

In \cite{TomamichelMonogamyGame2013}, security of one-sided device-independent quantum key distribution (DIQKD) BB84 \cite{BB84} was proven using a monogamy-of-entanglement game. 
In order to reduce an attack on the protocol to a monogamy-of-entanglement game, we consider an entanglement-based variant of the original BB84 protocol, which implies security of the latter \cite{BB84withoutBell}. 

We consider a simplified version omitting information reconciliation and privacy amplification, in which two parties, a referee (usually in the literature, Alice) and Bob want to establish a secret key, tolerating a measurement error $p_{err}$ per qubit:
\begin{itemize}
    \item The referee prepares $n$ EPR pairs, keeps one half and sends the other half to Bob from each EPR pair. Bob confirms he received the $n$ qubits. 
    \item The referee picks a random basis $x_i$, either computational or Hadamard, to measure each qubit, sends $x_0,\ldots,x_{n-1}$ to Bob and both measure, obtaining $a=a_0 \dots a_{n-1}$ and $b=b_0 \dots b_{n-1}$, respectively. 
\end{itemize}
The key is successfully generated, up to information reconciliation and privacy amplification, if $d_{H}(a,b)\leq n p_{err}$.

The security of this protocol for untrusted Bob's measurement device, which could behave maliciously, can be reduced to the monogamy-of-entanglement game where Alice's goal is to guess the value of the referee's raw key~$a$ \cite{TomamichelMonogamyGame2013}, where it is proven that the protocol can tolerate a noise up to 1.5\% asymptotically.  

Here we consider a similar approach by first fixing the error $p_{err}$ and then seeing what is the probability that Alice guesses $x$, with the condition that Bob's answers must be at a relative Hamming distance at most $p_{err}$ from $x$. This security can be reduced to the constrained extended non-local game $\mathcal{G}_{C_{p_{err}}}^{\text{A }n\text{ guess}}$ described in Section~\ref{section:Alice_guess_game}.

Obtaining the values of $\omega_{admiss}^{(k=1)}(\mathcal{G}_C^{n-QKD})$ for $n=1,2$, see Fig.~\ref{Fig:QKD}, we find nontrivial upper bounds for the winning probability  $\mathcal{G}_{C_{p_{err}}}^{\text{A }n\text{ guess}}$, which translate to security for QKD. Notice that if the error per qubit is at least $\frac{1}{2}$, the eavesdropper can use the strategy consisting of sharing $n$ EPR pairs with the referee, performing the same measurement, and manipulating Bob's device so that it answers random bits (which each are going to be correct with probability $\frac{1}{2}$), therefore, the eavesdropper's outcomes will be perfectly correlated with the referee's outcomes so that she guesses perfectly the raw key and for $p_{err}\geq\frac{1}{2}$,  $\omega_Q(\mathcal{G}_{C_{p_{err}}}^{\text{A }n\text{ guess}})=1$, this confirms the that the results for $p_{err}\geq\frac{1}{2}$ obtained via $\omega_{admiss}^{(k=1)}(\mathcal{G}_{C_{p_{err}}}^{\text{A }n\text{ guess}})$ are tight, see Fig.~\ref{Fig:QKD}.

In order to prove security for an arbitrary $n$ and a given laboratory error $p_{err}$, one would need to solve the corresponding SDP $\omega_{admiss}^{(k=1)}(\mathcal{G}_{C_{p_{err}}}^{\text{A }n\text{ guess}})$ (or larger $k$). However, the downside is that larger values of $n$ take considerably more computational resources.

\subsection{Application to quantum position verification}\label{sec:QPV}

Securely finding out a party's location (position verification) or writing messages that can only be read at a certain location are potentially-impactful tasks part of the field of \emph{position-based cryptography} (PBC). An important building block for PBC is so-called Position Verification (PV), where, in a 1-dimensional setup, an untrusted prover $P$ wants to convince verifiers $V_0,$ and $V_1$, placed left and right of $P$, respectively, that she is at a certain given position. In \cite{OriginalPositionBasedCryptChandran2009} it was proven that no secure classical protocol for position verification can exist, since there exists a general attack based on copying classical information. 

Due to the no-cloning theorem~\cite{Wootters1982NoCloning} the general classical attack does not apply if quantum information is used instead \cite{OriginalQPV_Kent2011,Malaney_QLP}, however, a general quantum attack exists which requires exponential entanglement~\cite{Buhrman_2014,Beigi_2011}. This means that hope for protocols secure against reasonable amounts of entanglement is alive, and indeed there has been much analysis on attacks on specific protocols~\cite{PatentKentANdOthers,OriginalQPV_Kent2011,Lau_2011,https://doi.org/10.48550/arxiv.1504.07171,Chakraborty_2015,speelman2016instantaneous,dolev2019constraining,dolev2022non,gonzales2019bounds,cree2022code}.
Also of note is recent work performing the first experiment that implements QPV in a lab setting~\cite{loffler2025towards}.

One of the simplest and best-studied QPV protocols~\cite{PatentKentANdOthers,OriginalQPV_Kent2011} is based on BB84 states, denoted by \QPVBB. The protocol consists of one verifier ($V_0$) sending a BB84 state ${\ket\psi\in\{\ket0,\ket1,\ket+,\ket-\}}$ to the prover, and the other verifier ($V_1$) sending classical information $x\in\{0,1\}$ describing in which basis the prover has to measure, either the computational basis ($x=0$) or the Hadamard basis ($x=1$). The prover then has to broadcast the measurement outcome $a\in\{0,1\}$ to both verifiers, with all communication happening at the speed of light.

The most general attack on any 1-dimensional QPV protocol is to place an adversary between $V_0$ and the prover, Alice, and another adversary between the prover and $V_1$,  Bob, such that the adversaries intercept the messages coming from their closest verifiers, have one round of simultaneous communication, and reply to their closest verifier.

For security analysis, we will consider the purified version of~\QPVBB, which is equivalent to the defined version above. In the purified version, $V_0$, instead of sending a BB84 state, prepares the EPR pair $\ket{\Phi^+}=\frac{\ket{00}+\ket{11}}{\sqrt{2}}$ and sends one qubit to $P$ and at some later point, $V_0$ measures their register in basis~$x$. In particular, it makes it easy to  delay the choice of basis~$x$ to a later point in time, which makes it clear that the attackers' actions are independent of the basis choice. 

In \cite{Buhrman_2014} it was proven that \QPVBB~is secure if the attackers do not pre-share entanglement before to intercept the information. They did so by upper bounding the probability that both Alice and Bob can guess the outcome of $V_0$'s measurement (which they use to answer in the last step of the attack). 
In~\cite{TomamichelMonogamyGame2013}, it was shown that security of \QPVBB~can be reduced to the $ \mathcal{G}^{BB84}$ game.

In an experimental implementation, one might expect the prover to answer incorrectly with a certain probability. However, since $P$ broadcasts classical information, it is not expected that she 
ever sends different answers (\emph{inconsistent answers}) to the two verifiers. Therefore, if in any round the verifiers receive an inconsistent answer, they can abort the entire protocol since they have observed something that will never happen in an honest execution of the protocol. Therefore, it is natural to study security adding related constraints to the attackers, which might lower their probability to win the extended non-local game, and thereby provide a tighter bound on the best attack of the protocol. A natural constraint to $\mathcal{G}^{BB84}$ would be to forbid inconsistent answers, i.e.\ to impose the set of constraints $C$ for every  strategy $S_{Q}=\{\rho_{RAB},A^x_a,B^y_b\}_{x,y,a,b}$ given by
\begin{equation}
    \sum_{x,a,a'\neq b'}q(x)\tr{(V^{x}_{a}\otimes A^x_{a'} \otimes B^x_{b'})\rho_{RAB}}=0.
\end{equation}
This corresponds to the constrained $\mathcal{G}^{BB84}$ game analyzed in Section~\ref{section:cons_BB84}, showing that unentangled attackers cannot win it with probability greater than $\cos^2(\frac{\pi}{8})$, corresponding to the optimal winning probability without constraining the game.

In a practical application, not only errors arise, but also a sizable fraction sent from the verifiers to the prover is lost. In the case of \QPVBB, if the transmission rate from the channel connecting $V_0$ and $P$ is $\eta$, the prover is expected to answer either 0, 1 or that she did not receive the qubit (with probability $1-\eta$). The lossy version of \QPVBB~will be denoted by \QPVBBeta. In \cite{Escol_Farr_s_2023} it was shown that the security of \QPVBBeta~could be reduced to the lossy-and-constrained extended non-local game $\mathcal{G}_{C,\eta}^{BB84}$, which we analyzed in Section \ref{section:lossy_moe_games}, with $\omega_{Q}(\mathcal{G}_{C,\eta}^{BB84})=\frac{1}{\sqrt{2}}\eta+\sin^2\big(\frac{\pi}{8}\big)$. We therefore see that for an honest prover without error, the protocol remains secure for $\eta>\frac{1}{2}$, since the honest prover will answer correctly with probability $\eta$ and $\eta>\frac{1}{\sqrt{2}}\eta+\sin^2\big(\frac{\pi}{8}\big)$ for all $\eta>\frac{1}{2}$. Moreover, for a given $\eta$, the protocol tolerates a total error $p_{err}$ as long as
\begin{equation}
    \eta(1-p_{err})>\frac{1}{\sqrt{2}}\eta+\sin^2\big(\frac{\pi}{8}\big).
\end{equation}
Since the results were tight, the above inequality provides the optimal relation with the error. We saw that using our new methodology,  the results of \cite{Escol_Farr_s_2023} are verified in a simplified way, since there it was necessary to derive inequalities using the norms of the verifier's measurements but just the description of the protocol. 

The application of Section \ref{section:lossy_moe_games} in QPV not only verifies in a different way preveously known results but improves some providing tight results. Consider the extension of \QPVBB~consisting on $V_0$ sending a uniformly random state from the set  $\{\ket0,\ket1,\ket+,\ket-,\ket{+i} ,\ket{-i}\}$ and $V_1$ sending in which basis to measure (Hadamard, computational or $\{\ket{\pm i}\}$), introduced in \cite{OriginalQPV_Kent2011}. The security of this protocol can be reduced to the winning probability of the extended non-local game $\mathcal G^{3-bases}$, as shown in \cite{Escol_Farr_s_2023}. This game was analyzed in Section \ref{section:lossy_moe_games}, and it has optimal value $\frac{1}{2}+\frac{\sqrt{3}}{6}$, therefore, attackers can spoof the verifiers with at most a probability $\frac{1}{2}+\frac{\sqrt{3}}{6}\simeq0.788675 $ per round. In addition, if loss is considered, security reduced to the winning probability of $\mathcal G^{3-bases}_{C,\eta}$ where the set of constraints is such that different answers are forbidden, whose optimal winning probability is given in Result \ref{result:3_bases}. Notice that we have that $\omega_Q(\mathcal G^{3-bases}_{C,\eta})=\omega_Q(\mathcal G^{3-bases}_{\eta})$, since the optimal strategy given in Result \ref{result:3_bases} is such that $C$ is fulfilled. Then, for an error-free prover, one expects correct answers with probability $\eta$ and the protocol is secure for $\eta>\frac{1}{3}$, since for this range, $\eta>\omega_Q(\mathcal G^{3-bases}_{C,\eta})$. Then, given a transmission rate $\eta$, the tight relation with the error $p_{err}$ is given by:
\begin{enumerate}[label=\textit{(\roman*)}]
    \item if $ \eta\in[\frac{1}{3},\frac{2}{3})$, 
    \begin{equation}
    \eta(1-p_{err})>\frac{1}{3}+\frac{1}{\sqrt{2}}(\eta-\frac{1}{3}),  
\end{equation}
\item if $\eta\in[\frac{2}{3},1]$,
\begin{equation}
    \eta(1-p_{err})>\left(\frac{1}{\sqrt{2}}-\frac{1}{\sqrt{3}}\right)+\frac{1-\sqrt{2}+\sqrt{3}}{2}\eta.
\end{equation}

\end{enumerate}

\subsubsection{Security for 2-fold parallel repetition of \QPVBBeta}

The security of the $n$-fold parallel repetition of \QPVBBeta, given that the attackers do not pre-share entanglement, can can be reduced to the $n$-fold parallel repetition of the lossy extended non-local game $\mathcal{G}_{\eta}^{BB84\times n}$, with
$\mathcal X=\mathcal Y=\{0,1\}^n$ and $\mathcal A=\mathcal B=\{0,1,\perp\}^n$.

In this section, we analyze the case $n=2$ and we show security for the 2-fold parallel repetition of \QPVBBeta, given that the attackers do not pre-share entanglement. Imposing the same loss rate for both qubits, i.e.,
\begin{equation}
        \sum_{x_0x_1}q(x_0x_1)\tr{(\mathbb{I}_R\otimes\mathbb{I}_R\otimes A^{x_0x_1}_{\perp\perp} B^{x_0x_1}_{\perp\perp}) \rho_{RAB}}=(1-\eta)^2,
\end{equation}
for all $a_0,a_0'\in\{0,1\}$,
\begin{equation}
        \sum_{x_0x_1}q(x_0x_1)\tr{(\mathbb{I}_R\otimes\mathbb{I}_R\otimes A^{x_0x_1}_{a_0\perp} B^{x_0x_1}_{a_0'\perp}) \rho_{RAB}}=\frac{1}{2}\eta(1-\eta),
\end{equation}
and for all $a_1,a_1'\in\{0,1\}$,
\begin{equation}
        \sum_{x_0x_1}q(x_0x_1)\tr{(\mathbb{I}_R\otimes\mathbb{I}_R\otimes A^{x_0x_1}_{\perp a_1} B^{x_0x_1}_{\perp a_1'}) \rho_{RAB}}=\frac{1}{2}\eta(1-\eta), 
\end{equation}
we obtain an upper bound solving the SDP that gives the value $\omega_{admiss}^{(k=1)}( \mathcal{G}_{\eta}^{BB84\times 2})$. The obtained values, see \cite{code} for the code, are plotted in Fig.~\ref{Fig:parallel_QPV}, together with the winning value of the parallel repetition of the optimal strategy for $ \mathcal{G}_{\eta}^{BB84}$. We see that the upper bounds are slightly grater than the values corresponding to the parallel repetition of the optimal strategy, and thus it remains open if strong parallel repetition holds when there is loss. 

An error-free honest prover will answer correctly with probability $\eta^2$, which is strictly larger than the upper bounds found for $\omega_Q(\mathcal{G}_{\eta}^{BB84\times 2}$), for all $\eta>\frac{1}{2}$. In addition,   if the the total error of the prover is $p_{err}$, we have security for 2-fold parallel repetition as long as the probability of answering correctly, given by $\eta^2(1-p_{err})$, is such that
\begin{equation}
    \eta^2(1-p_{err})>\omega_{admiss}^{(k=1)}( \mathcal{G}_{\eta}^{BB84\times 2}).
\end{equation}

Security for the $n$-fold parallel repetition for a given $n$ can be computed analogously. Since our techniques do run into a practical limit when increasing $n$, we leave the open problem to create a more efficient program that computes the parallel repetition faster. 

\begin{figure}[H]
\centering
\includegraphics[width=125mm]{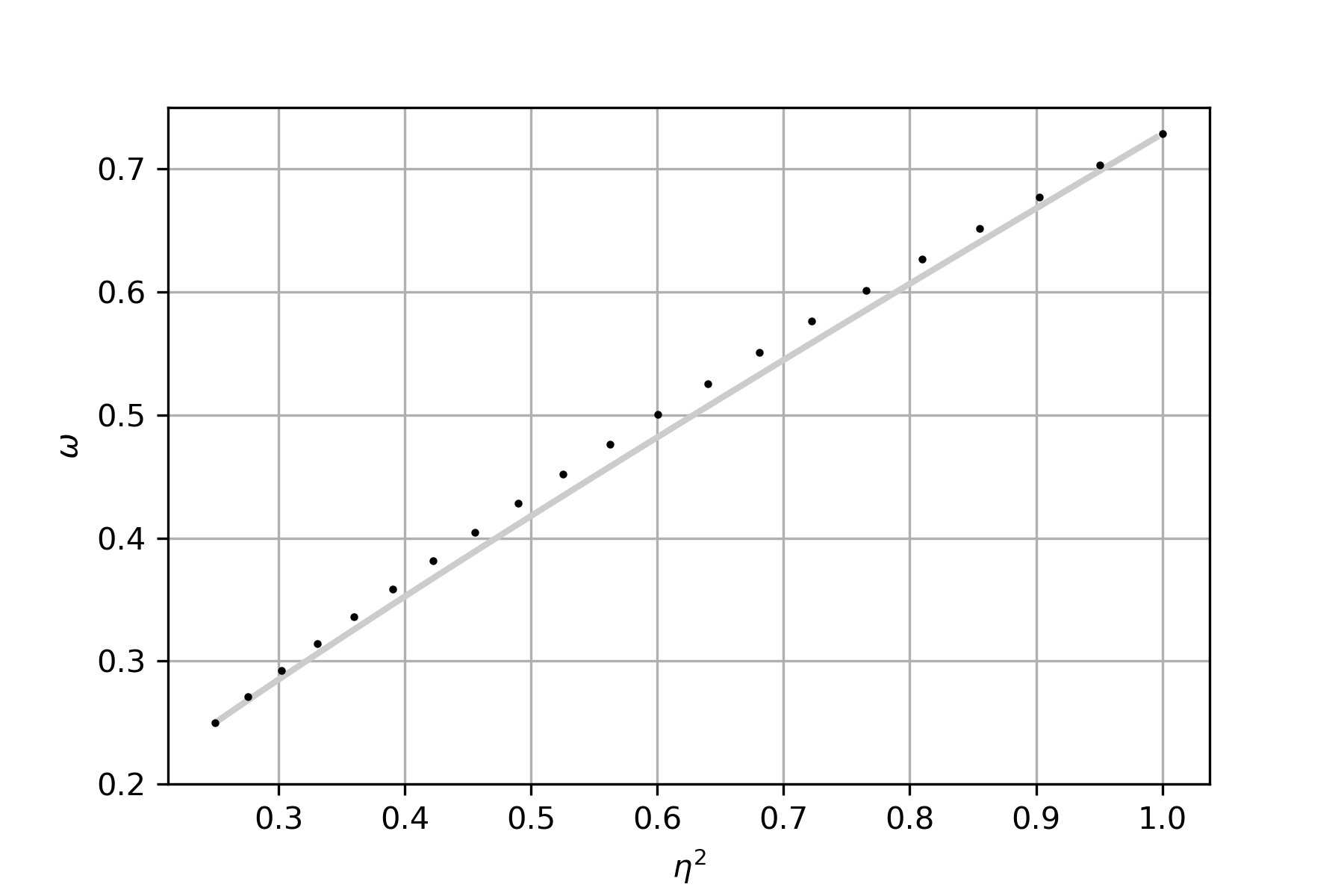}
\caption{Values of $\omega_{admiss}^{(k=1)}( \mathcal{G}_{\eta}^{BB84\times 2})$ (black dots) and winning values for the parallel repetition of the optimal strategy for $\mathcal{G}_{\eta}^{BB84}$ (gray continuous line).}
\label{Fig:parallel_QPV}
\end{figure}

\section{Discussion and open problems}

We give an approach using SDP hierarchies to extend security proofs of quantum cryptographic protocols to realistic scenarios. A first natural extension of this work, could be to apply our techniques to other problems in quantum cryptography, especially when proving security of near-term experimental setups.
A possible future task could be optimizing the code provided in this work, making it feasible to compute SDP solutions for larger instances of $n$ in some of the above analyzed games. In addition, since the solutions of the solved SDPs provide tighter results than Proposition~\ref{prop:bound_local_guess}, it would be interesting to find either a tighter version of that proposition or the optimal value for all $n$. 
In our application to relativistic bit commitment, we show partial loss tolerance of the protocol without having to adjust the communication pattern. A different way to make the protocol loss tolerant could be to introduce an extra round of communication, where the committer announces immediately which qubits arrived successfully.

\section{Acknowledgments }
This work was supported by the Dutch Ministry of Economic Affairs and Climate Policy (EZK), as part of the Quantum Delta NL programme. We would like to thank an anonymous referee for pointing out a mistake in the SDP solutions of Section~\ref{section:cons_BB84} in an earlier version of this work --- thanks to their observation we were able to correct a typo in the code for the corresponding game \cite{code}, and could also find strategies that matched the bounds provided by the (correct) SDP solutions.

\bibliographystyle{quantum}
\bibliography{references}

\begin{appendices}
    \section{Proof of Proposition \ref{prop:bound_local_guess}}
We show the proof of Proposition \ref{prop:bound_local_guess} \cite{Pital_a_Garc_a_2018}. For that, we need the following Definition and Lemmas.
\begin{definition} Let $n\in\mathbb N$. Two permutations $\pi,\pi':[n]\rightarrow[n]$ are said to be \emph{orthogonal} if $\pi(i)\neq\pi'(i)$ for all $i\in[n]$. 
\end{definition}

\begin{lemma} \label{lem:sum_projectors}\emph{(Lemma 2 in \cite{TomamichelMonogamyGame2013})}
    Let $\Pi^1,\ldots,\Pi^n$ be projectors acting on a Hilbert space $\mathcal H$. Let $\{\pi_k\}_{k\in[n]}$ be a set of mutually orthogonal permutations. Then, 
    \begin{equation}
        \bigg\|\sum_{i\in[n]}\Pi^i \bigg\|\leq \sum_{k\in[n]}\max_{i\in[n]}\big\| \Pi^i\Pi^{\pi_k(i)}\big\|.
    \end{equation}
\end{lemma}

\begin{remark}
    There always exists a set of $n$ permutations of $[n]$ that are mutually orthogonal, an example is the ~$n$ cyclic shifts.  
\end{remark}

\begin{lemma}\label{lem:norm_product}(Lemma 1 in \cite{TomamichelMonogamyGame2013})
    Let $A,B,L\in\mathcal{B}(\mathcal{H})$ such that $AA^\dagger\succeq B^\dagger B$. Then it holds that $\norm{AL}\geq\norm{BL}$. 
\end{lemma}

Now, we are in position to prove Proposition \ref{prop:bound_local_guess}, which uses similar ideas as in the proof of Theorem 3 in \cite{TomamichelMonogamyGame2013}. Let $\mathcal{S}=\{\rho_{RAB}, A^x_a, B^x_b\}$ be an arbitrary strategy.  As above, let $\mathcal{T}_{xz}:=\{i\mid z_i=x\}$, and denote by $\mathcal{T}_{xz}^c$ its complement, and, abusing of notation, we will write $\sum_{v\in\mathcal{T}_{xy}}$ for a bit string $v\in\{0,1\}$  to sum the indices $i$ of $v$ such that $i\in \mathcal{T}_{xz}$.     Then
\begin{equation}\label{eq:w<=permutations}
\begin{split}
    \omega_{Q}&(\mathcal{G}^{n-\text{local guessing}})\big|_{\mathcal{S}}=\frac{1}{2^{n+1}}\sum_{(z,x,v,a,b)\in W}\tr{(V^z_v\otimes A^x_a\otimes B^x_b)\rho_{RAB}}\\&=\frac{1}{2^{n+1}}\sum_{z,x}\sum_{v\in\mathcal T_{zx}^c,a}\tr{\big(\ketbra{a^{x\ldots x}_{\mathcal{T}_{zx}}}{a^{x\ldots x}_{\mathcal{T}_{zx}}}\otimes \ketbra{v^{1-x\ldots 1-x}_{\mathcal{T}_{zx}^c}}{v^{1-x\ldots 1-x}_{\mathcal{T}_{zx}^c}}\otimes A^x_{a_{\mathcal{T}_{zx}}a_{\mathcal{T}_{zx}^c}}\otimes B^x_{a_{\mathcal{T}_{zx}}a_{\mathcal{T}_{zx}^c}}\big)\rho_{RAB}}\\&\leq
    \frac{1}{2^{n+1}}\bigg\|\sum_{z,x}\sum_{v\in\mathcal T_{zx}^c,a}\ketbra{a^{x\ldots x}_{\mathcal{T}_{zx}}}{a^{x\ldots x}_{\mathcal{T}_{zx}}}\otimes \ketbra{v^{1-x\ldots 1-x}_{\mathcal{T}_{zx}^c}}{v^{1-x\ldots 1-x}_{\mathcal{T}_{zx}^c}}\otimes A^x_{a_{\mathcal{T}_{zx}}a_{\mathcal{T}_{zx}^c}}\otimes B^x_{a_{\mathcal{T}_{zx}}a_{\mathcal{T}_{zx}^c}}\bigg\|\\&=\frac{1}{2^{n+1}}\bigg\|\sum_{z,x}\sum_{v\in\mathcal T_{zx}^c,a}\ketbra{a^{x\ldots x}_{\mathcal{T}_{zx}}}{a^{x\ldots x}_{\mathcal{T}_{zx}}}\otimes \mathbb I_{\mathcal T^c_{zx}}\otimes A^x_{a_{\mathcal{T}_{zx}}a_{\mathcal{T}_{zx}^c}}\otimes B^x_{a_{\mathcal{T}_{zx}}a_{\mathcal{T}_{zx}^c}}\bigg\|\\&=\frac{1}{2^{n+1}}\bigg\|\sum_{z,x}\Pi^{zx}\bigg\|\leq \frac{1}{2^{n+1}}\sum_k\max_{xz}\norm{\Pi^{zx}\Pi^{\pi_k(zx)}},
    \end{split}
\end{equation}
where we used $\sum_{v\in\mathcal T_{zx}^c}\ketbra{v^{1-x\ldots 1-x}_{\mathcal{T}_{zx}^c}}{v^{1-x\ldots 1-x}_{\mathcal{T}_{zx}^c}}=\mathbb I_{\mathcal T^c_{zx}}$ where $I_{\mathcal T^c_{zx}}$ is the identity in the systems $T^c_{zx}$, in the last equality we defined $\Pi^{zx}:=\sum_{a}\ketbra{a^{x\ldots x}_{\mathcal{T}_{zx}}}{a^{x\ldots x}_{\mathcal{T}_{zx}}}\otimes \mathbb I_{\mathcal T^c_{zx}}\otimes A^x_{a_{\mathcal{T}_{zx}}a_{\mathcal{T}_{zx}^c}}\otimes B^x_{a_{\mathcal{T}_{zx}}a_{\mathcal{T}_{zx}^c}}$, and we used Lemma \ref{lem:sum_projectors} in the last inequality, for  a set of mutually orthogonal permutations $\{\pi_k\}$. Fix $k$, let $z'x':=\pi_k(zx)$, let $\mathcal I:=\mathcal{T}_{zx}\cap\mathcal{T}_{z'x'}$, let $i=\abs{\mathcal I}$ and denote $\mathcal I^c=[n]\setminus \mathcal I$. We define
\begin{equation}
\begin{split}
    P^{zx}&:=\sum_{a}\ketbra{a_{\mathcal I}^{x\ldots x}}{a_{\mathcal I}^{x\ldots x}}\otimes \mathbb I_{\mathcal I^c}\otimes A^x_a\otimes\mathbb I_{B}, \\Q^{z'x'}&:=\sum_{a'}\ketbra{a_{\mathcal I}'^{x'\ldots x'}}{a_{\mathcal I}'^{x'\ldots x'}}\otimes \mathbb I_{\mathcal I^c}\otimes \mathbb I_A\otimes B^x_{a'}.
\end{split}
\end{equation}
Note that $\Pi^{zx}\preceq P^{zx}$ and $\Pi^{z'x'}\preceq Q^{z'x'}$, then, by Lemma \ref{lem:norm_product}, $\norm{\Pi^{zx}\Pi^{z'x'}}\leq\norm{P^{zx}Q^{z'x'}}=\norm{P^{zx}Q^{z'x'}Q^{z'x'}P^{zx}}^{\frac{1}{2}}=\norm{P^{zx}Q^{z'x'}P^{zx}}^{\frac{1}{2}}$. We compute the last product of matrices:
\begin{equation}
\begin{split}
     P^{zx}Q^{z'x'}P^{zx}&=\sum_{a,a',a''}\ket{a^{x\ldots x}_{\mathcal I}}\braket{a_{\mathcal I}^{x\ldots x}}{a_{\mathcal I}'^{x'\ldots x'}}\braket{a_{\mathcal I}'^{x'\ldots x'}}{a_{\mathcal I}''^{x\ldots x}}\bra{a_{\mathcal I}''^{x\ldots x}}\otimes \mathbb I_{\mathcal I^c}\otimes A^x_{a}A^{x}_{a''}\otimes B^{x'}_{a'}\\
     &=\sum_{a,a'}\abs{\braket{a_{\mathcal I}^{x\ldots x}}{a_{\mathcal I}'^{x'\ldots x'}}}^2\ketbra{a^{x\ldots x}_{\mathcal I}}{a^{x\ldots x}_{\mathcal I}}
     \otimes \mathbb I_{{\mathcal I^c}}\otimes A^x_{a}\otimes B^{x'}_{a'},
\end{split}
\end{equation}
where we used that $A^x_aA^x_{a'}=\delta_{aa'}A^x_a$. We distinguish two cases:
\begin{enumerate}[label=\textit{(\roman*)}]
    \item if $x\neq x'$, we have that $ \abs{\braket{a_{\mathcal I}^{x\ldots x}}{a_{\mathcal I}'^{x'\ldots x'}}}^2=2^{-i}$,  and then, 
    \begin{equation*}
        P^{zx}Q^{z'x'}P^{zx}=\sum_{a}2^{-i}\ketbra{a^{x\ldots x}_{\mathcal I}}{a^{x\ldots x}_{\mathcal I}}
     \otimes \mathbb I_{{\mathcal I^c}}\otimes A^x_{a}\otimes\sum_{a'} B^{x'}_{a'}=\sum_{a}2^{-i}\ketbra{a^{x\ldots x}_{\mathcal I}}{a^{x\ldots x}_{\mathcal I}}
     \otimes \mathbb I_{{\mathcal I^c}}\otimes A^x_{a}\otimes \mathbb I_B,
    \end{equation*}
    therefore 
    \begin{equation}
        \norm{P^{zx}Q^{z'x'}P^{zx}}\leq 2^{-i}.
    \end{equation}
    \item If $x= x'$, then $\abs{\braket{a_{\mathcal I}^{x\ldots x}}{a_{\mathcal I}'^{x'\ldots x'}}}^2=\delta_{aa'}$ and we use the trivial bound 
  \begin{equation}
      \norm{P^{zx}Q^{z'x'}P^{zx}}\leq 1.
  \end{equation}
\end{enumerate}
Therefore, we have that 
\begin{equation}
    \norm{\Pi^{zx}\Pi^{z'x'}}\leq\begin{cases}
        \sqrt{2^{-i}} &\text{ if } x\neq x',\\
        1 &\text{ if } x= x'.
    \end{cases}
\end{equation}
In order to apply it in the bound \eqref{eq:w<=permutations}, consider the set of permutations given by $\pi_k(zx)=zx\oplus k$, where $zx, k\in \{0,1\}^{n+1}$ (they are such that they have the same Hamming distance). Notice that exactly half of the permutations (i.e.~$2^n$) are such that $x=x'$ and the other half are such that $x\neq x'$. For the latter case, there are $\binom{n}{i}$ permutations with Hamming distance $i$. Then, 
\begin{equation}
\begin{split}
     \omega_{Q}(\mathcal{G}^{n-\text{local guessing}})\big|_{\mathcal{S}}&\leq\frac{1}{2^{n+1}}\bigg\|\sum_{z,x}\Pi^{zx}\bigg\|\leq \frac{1}{2^{n+1}}\sum_k\max_{xz}\norm{\Pi^{zx}\Pi^{\pi_k(zx)}}\\&\leq \frac{1}{2^{n+1}}\left(2^n+\sum_{i=0}^n\binom{n}{i} \sqrt{2^{-i}}\right)=\frac{1}{2}+\frac{1}{2}\left(\frac{1+1/\sqrt{2}}{2}\right)^n.
     \end{split}
\end{equation}
\end{appendices}

\end{document}